\documentclass[11pt]{article}
\usepackage[left=1in, top=1in, right=1in, bottom=1in]{geometry}
\usepackage{amsmath, amssymb, amsthm, amsfonts}
\usepackage{mathabx}
\usepackage[numbers]{natbib}
\usepackage{float}
\usepackage{caption}
\usepackage{xfrac}
\usepackage[T1]{fontenc}
\usepackage{diagbox}
\usepackage{bbm}
\usepackage{multirow}
 
\usepackage{booktabs} 
\usepackage[ruled]{algorithm2e} 
\usepackage{algorithmic}

\SetAlFnt{\small}
\SetAlCapFnt{\small}
\SetAlCapNameFnt{\small}
\SetAlCapHSkip{0pt}
\IncMargin{-\parindent}

\newtheorem{theorem}{Theorem}[section]

\newtheorem{lemma}[theorem]{Lemma}

\newtheorem{claim}[theorem]{Claim}

\usepackage{colortbl}

\definecolor{amethyst}{rgb}{0.6, 0.4, 0.8}

\newcommand{\sfe}[2]{{#1} = E_{#2}}

\newcommand{\aat}[2]{{#1} = e_{#2}}
\newcommand{\indic}{\mathbbm{1}}

\newcommand{\pair}[2]{\mu_{#1}(#2)}

\newcommand{\matat}[2]{\text{matched}({#1},@{#2})}
\newcommand{\matbefore}[2]{\text{matched}({#1},<{#2})}
\newcommand{\matby}[2]{\text{matched}({#1},\leq{#2})}
\newcommand{\avail}[2]{#1 \text{ available }@{#2}}
\newcommand{\wet}{\weight_{e_t}}
\newcommand{\wmt}{\weight(\mu_t)}
\renewcommand{\citet}[1]{\cite{#1}} 

\usepackage{lineno}

\newcommand{\be}{\begin{equation}}
\newcommand{\ee}{\end{equation}}
\newcommand{\beq}{\begin{equation*}}
\newcommand{\eeq}{\end{equation*}}

\newcommand{\eps}{\varepsilon}

\newcommand{\AutoAdjust}[3]{\mathchoice{ \left #1 #2  \right #3}{#1 #2 #3}{#1 #2 #3}{#1 #2 #3} }
\newcommand{\Xcomment}[1]{{}}

\newcommand{\InBrackets}[1]{\AutoAdjust{[}{#1}{]}}
\newcommand{\Ex}[2][]{\operatorname{\mathbf E}_{#1}\InBrackets{#2}}

\newcommand{\Prx}[2][]{\operatorname{\mathbf{Pr}}_{#1}\InBrackets{#2}}

\newcommand{\eqdef}{\overset{\mathrm{def}}{=\mathrel{\mkern-3mu}=}}
\newcommand{\vect}[1]{\ensuremath{\mathbf{#1}}}

\newcommand{\RN}[1]{%
  \textup{\uppercase\expandafter{\romannumeral#1}}%
}
\newcommand{\fst}{\RN{1}}
\newcommand{\snd}{\RN{2}}

\newcommand\restr[2]{{
  \left.\kern-\nulldelimiterspace 
  #1 
  \vphantom{\big|} 
  \right|_{#2} 
  }}
\def\prob{\Prx}

\newcommand{\weight}{w}
\newcommand{\weights}{\vect{w}}

\newcommand{\alg}{\textsf{ALG}}
\newcommand{\opt}{\textsf{OPT}}

\newcommand{\dhistory}{\mathcal{DH}}

\renewcommand{\emptyset}{\varnothing}

\newcommand{\dd}{\mathrm{d}}

\def \reals {{\mathbb R}}
\def \naturals {{\mathbb N}}
\newcommand{\subsets}[2][\naturals]{\binom{#1}{#2}}

\newcommand{\bvt}{\widetilde{V}}

\newcommand{\val}{\lambda}
\newcommand{\Val}{\Lambda}

\title{Secretary Matching with General Arrivals
	\thanks{This work is supported by Science and Technology Innovation 2030 - ``New Generation of Artificial Intelligence'' Major Project No.(2018AAA0100903), Innovation Program of Shanghai Municipal Education Commission, Program for Innovative Research Team of Shanghai University of Finance and Economics (IRTSHUFE) and the Fundamental Research Funds for the Central Universities. The first two authors are partially supported by the European Research Council (ERC) under the European Union's Horizon 2020 research and innovation program (grant agreement No. 866132), and by the Israel Science Foundation (grant number 317/17).}
}

\author{Tomer Ezra\thanks{Tel Aviv University. Email:  \texttt{tomer.ezra@gmail.com}, \texttt{michal.feldman@cs.tau.ac.il }}
	\and
	Michal Feldman\footnotemark[1]
	\and
	Nick Gravin\thanks{ITCS, Shanghai University of Finance and Economics. Email: \texttt{\{nikolai, tang.zhihao\}@mail.shufe.edu.cn}} \and Zhihao Gavin Tang\footnotemark[2] 
}

\begin{document}
	\date{}
	\maketitle
	
\begin{abstract}
	We provide online algorithms for {\em secretary matching} in general weighted graphs, under the well-studied models of vertex and edge arrivals. In both models, edges are associated with arbitrary weights that are unknown from the outset, and are revealed online. 
Under vertex arrival, vertices arrive online in a uniformly random order; upon the arrival of a vertex $v$, the weights of edges from $v$ to all previously arriving vertices are revealed, and the algorithm decides which of these edges, if any, to include in the matching. 
Under edge arrival, edges arrive online in a uniformly random order; upon the arrival of an edge $e$, its weight is revealed, and the algorithm decides whether to include it in the matching or not. 
We provide a $5/12$-competitive algorithm for vertex arrival, and show it is tight.
For edge arrival, we provide a $1/4$-competitive algorithm.
Both results improve upon state of the art bounds for the corresponding settings. 
Interestingly, for vertex arrival, secretary matching in general graphs outperforms secretary matching in bipartite graphs with 1-sided arrival, where $1/e$ is the best possible guarantee.

\end{abstract}

\section{Introduction}
\label{sec:intro}

A common tension in market scenarios, faced over and over again by individuals and firms, is choosing the right timing to commit to a decision.
This tension arises when one chooses their life-long partner, makes a reservation in Airbnb, accepts a job offer, or any other scenario where a decision should be made in the present, without knowing whether or to what extent a better option would arrive in the future.

The most basic mathematical model of such scenarios has been studied in the mathematical literature of optimal stopping theory. 
In a stopping problem, there are $n$ rounds, and a sequence of $n$ values $w_1, \ldots, w_n$ that are unknown from the outset. 
In every round $t$, the value $w_t$ is revealed, and the decision maker makes an irrevocable decision whether to select $w_t$, in which case the process ends with value $w_t$, or continue to the next round (unless it's round $n$), in which case the value $w_t$ is lost forever and the process continues to round $t+1$.
The goal is to maximize the obtained value.
The {\em competitive ratio} of an algorithm $\alg$ is the minimum ratio between the expected value obtained by $\alg$ and the globally maximal value, over all value sequences $\weights=(w_1,\ldots,w_n)$.

One can verify that not much can be done if the process is entirely adversarial. 
Two alternative models of stochastic variants have been studied extensively in the literature, and become to be known as the {\em secretary problem} \cite{Gardner60,ferguson1989} and the {\em prophet inequality} \cite{krengelS77,krengel1978semiamarts}.
In the secretary problem, the value sequence $\weights$ is arbitrary, but the values are assumed to arrive in a uniformly random order.
The secretary problem is known to admit a competitive ratio of $1/e$, and this is the best possible ratio \cite{Dynkin1963TheOC}.
In the prophet setting, every value $w_t$ is drawn (independently) from a probability distribution that is known from the outset. 
The prophet problem is known to admit a competitive ratio of $1/2$, and this is the best possible ratio \cite{krengelS77,krengel1978semiamarts,samuel1984comparison}.

A natural question arises: do these results extend to more complex stochastic optimization problems? This problem received a lot of attention in recent years in combinatorial structures such as uniform matroids \cite{HajiaghayiKS07,Kleinberg05}, graphical matroids \cite{KorulaP09}, general matroids  \cite{BabaioffIKK18,KleinbergW19}, intersection of matroids \cite{KleinbergW19}, matching in graphs \cite{KorulaP09,KesselheimRTV13,GravinTW19,EzraFGT20}, and general downward-closed feasibility constraints \cite{rubinstein2016beyond}.


Of particular interest to this paper is the extension to {\em matching} problems in weighted graphs, where the goal is to select a matching of maximum weight. 
Matching problems have been of great interest in the last decade, partly due to their high applicability to Internet markets \cite{Mehta13}, such as as online ad auctions, ridesharing platforms, online labor markets, and exchange markets for pairwise kidney exchange.

Two natural arrival models have been studied for matching, namely {\em vertex arrival} and {\em edge arrival}.
Edge arrival is the more standard one, in the sense that elements arrive one by one, as in classic secretary and prophet settings (see, e.g., \citet{KorulaP09} on graphical matroids).
However, in matching applications, the vertex arrival model is extremely natural, as the arriving entities often correspond to the vertices. 
In vertex arrival, vertices arrive one by one, each one along with its edges to all previous vertices. Various models of vertex arrival have been studied in the literature, including 1-sided vertex arrival in bipartite matching and general vertex arrival \cite{GamlathKMSW19}.
\paragraph{Online matching with vertex arrival.}
%
\citet{KorulaP09} introduced the following 1-sided bipartite matching setting, modeled as an online matching problem in a weighted bipartite graph $G=(L,R;E)$. 
A pool of jobs are available in the market (associated with vertices in $R$). 
Potential employees (associated with vertices in $L$) arrive one by one, in a random order. 
Upon the arrival of a potential employee her value for every job is revealed (the weights on the corresponding edges), and the algorithm should either match the employee to one of the available jobs or leave it unmatched. The goal is to maximize the total value in the market.
The special case where there is a single job coincides with the classic secretary problem, thus the $1/e$ is the limit of what can be achieved.
But the decision making process under the matching scenario is much more complex, as the algorithm should decide not only whether or not to match it, but also to whom, among all available jobs.
\citet{KorulaP09} showed a $1/8$-competitive algorithm for this setting, and \citet{KesselheimRTV13} settled this problem by showing that the $1/e$ guarantee of the classic secretary problem extends to 1-sided bipartite matching.


The analogous 1-sided bipartite matching setting in the prophet model has been studied by \citet{FeldmanGL15}. 
Here too, the $1/2$ guarantee from the classic prophet setting extends to the more complex 1-sided bipartite matching. 

However, the underlying structure of 1-sided bipartite matching is quite restricted, and does not capture more dynamic scenarios, where vertices from both sides of the market arrive dynamically, e.g., passengers and drivers, items and buyers, jobs and employees. Also note that some scenarios cannot be captured by a bipartite graph at all, e.g., a pool of students who should be paired into roommates, or exchange markets for pairwise kidney exchange. Such scenarios are best captured by matching in a general graph, where upon the arrival of a vertex $v$, the weights on edges $(v,u)$ are revealed, for all previously arriving vertices $u$. 
The prophet version of this scenario has been studied by \citet{EzraFGT20}, who showed that the guarantee of $1/2$ extends even to this general matching setting (albeit using a different proof approach). 

That is, in the prophet model, the guarantee of $1/2$ obtained for the simplest setting extends all the way to matching in general graphs.
It is only natural to ask whether the same extension holds in the secretary setting as well. 

\noindent {\bf Main Question 1:} What is the competitive ratio for online matching problem in general graphs in the secretary setting with vertex arrival? Can we achieve the $1/e$ guarantee that holds for 1-sided bipartite matching? 

Our answer is: we can achieve a better guarantee! Indeed, the impossibility result of the classic secretary problem does not apply. Due to random arrival order, the single vertex in the $R$ side of the graph may come late enough so that many of the edge weights would be revealed simultaneously, enabling us to break the $1/e$ barrier.

\noindent {\bf Theorem:} Matching secretary in general graphs with vertex arrival admits a $5/12$-competitive algorithm.
Moreover, $5/12$ is the best possible competitive ratio. 

\noindent {\bf Remark:} For ordinal setting (i.e., if the algorithm is based only on pairwise comparisons of edges without observing associated numerical values) our result implies a competitive ratio of $\frac{5}{24}$. 

\paragraph{Online matching with edge arrival.}
We now turn to online matching with edge arrival. 
In this model, the edges arrive one by one. Upon the arrival of an edge, its weight is revealed, and if both its endpoints are available, the algorithm makes an irrevocable decision whether to include it in the matching. 
This model has been studied in both the prophet and secretaries model, but unlike the vertex arrival model, no tight bounds are known. 
For prophet, the competitive ratio is known to lie between $0.337$ and $3/7$ \cite{GravinW19,EzraFGT20,Tristanthesis}.
For secretary, \citet{KesselheimRTV13} have established a competitive ratio of $1/(2e)$, by a reduction from edge arrival to vertex arrival in hypergraphs.
The upper bound of $1/e$ from the classic secretary setting applies here. Shrinking the gap between $1/(2e)$ and $1/e$ is our second main problem.

\noindent {\bf Main Question 2:} What is the competitive ratio for online matching problem in general graphs in the secretary setting with edge arrival? Can we achieve a better guarantee than $1/(2e)$?

We answer this question in the affirmative.

\noindent {\bf Theorem:} Matching secretary in general graphs with edge arrival admits a $1/4$-competitive algorithm.

The design and analysis of our algorithm for edge arrival carry over to the more general online bipartite hypergraph secretary matching problem (with suitable adjustments). 
Algorithm~\ref{alg:hypergraph} in Appendix~\ref{sec:hypergraph} gives a competitive ratio of $1/d^{\frac{d}{d-1}}$, where $d+1$ equals the maximum size of the hyperedges. This improves (by a constant factor) upon the previous lower bound of $\frac{1}{ed}$ \cite{KesselheimRTV13}.

\begin{table}[]
	\centering
	\begin{tabular}{|l|l|l|l|l|}
		\hline
		&                                                                                        &    &    Secretary                   & Prophet                                                                                                     \\ \hline
		\multirow{4}{*}{Vertex arrival} & \multirow{2}{*}{\begin{tabular}[c]{@{}l@{}}1-sided \\ bipartite matching\end{tabular}} & LB & $\geq 1/e$ \cite{KesselheimRTV13} & $\geq 1/2$ \cite{FeldmanGL15}                                                                          \\ \cline{3-5} 
		&                                                                                        & UB &      $\leq 1/e$       & $\leq 1/2$                                                                                          \\ \cline{2-5} 
		& \multirow{2}{*}{General graphs}                                                        & LB &  $\geq 5/12$   [Theorem~\ref{th:5-over-12}]  & $\geq 1/2$ \cite{EzraFGT20}                                                            \\ \cline{3-5} 
		&                                                                                        & UB &    $\leq 5/12$   [Theorem~\ref{thm:upper-bound-vertex}]          & $\leq 1/2$                                                                              \\ \hline
		\multirow{2}{*}{Edge arrival}   & \multirow{2}{*}{}                                                                      & LB & \begin{tabular}[c]{@{}l@{}}$\geq 1/(2e)$ \cite{KesselheimRTV13}\\ $\geq 1/4$  [Theorem~\ref{thm:edge_arrival}]
		\end{tabular}       &   $\geq 0.337$ \cite{EzraFGT20}\\ \cline{3-5} 
		&                                                                                        & UB &     $\leq 1/e$       & $\leq 4/9$ \cite{GravinW19}                                                                                           \\ \hline
	\end{tabular}
	\caption{Our results and previous results. UB and LB refer to upper and lower bounds, respectively.}
	\label{tab:results}
\end{table}

\subsection{Our Techniques}
\label{sec:tech}

\paragraph{Vertex arrival: lower bound.}
At a high level, our algorithm follows the standard explore \& exploit approach for secretary problems, which has been adopted by \citet{KesselheimRTV13} for matching with 1-sided vertex arrival.
The algorithm begins with an exploration phase, where no matches are made. 
Then, in the exploitation phase, it finds at each step the optimum matching over the set of vertices that already arrived, and matches the latest vertex to its partner in the optimal matching whenever possible.
Our algorithm departs from the one used in \cite{KesselheimRTV13} in two ways: (1) it uses a longer exploration phase (half of the vertices instead of $1/e$ fraction), (ii), it ensures that the latest vertex {\em always} has a partner in the optimal matching (in particular, we treat differently the cases of odd and even number of vertices). 

The key feature of our analysis is a precise accounting of the probability that a given vertex is matched at every step of the algorithm. This accounting turns out to be more challenging in our general vertex arrival model than in 1-sided vertex arrival~\cite{KesselheimRTV13}. Indeed, in 1-sided vertex arrival, one has to do the accounting only for vertices in the {\em offline side} of the graph, and an upper bound on the matching probability is sufficient. In general graphs, on the other hand, every vertex can be either actively matched (i.e., matched upon its arrival), or passively matched (i.e., chosen as a partner of a vertex that arrives later), and there are non-monotone dependencies between matching probabilities of the vertices. That is, the event that a given vertex is matched may be either positively or negatively correlated with the event that another vertex is matched, depending on the set of vertices that have already arrived. For this reason, an upper bound on the probability of matching is not sufficient, and we have to make exact calculations of the probability that a vertex is matched in every step of the algorithm. To this end, our algorithm ensures that the optimum matching is a perfect matching in every step of the algorithm.


The same high level analysis approach is also pertinent in the cousin prophet inequality setting. Specifically, the Online Contention Resolution Schemes (OCRS) for matching in general graphs in a prophet setting~\cite{EzraFGT20} ensures that each vertex is matched to its realized partner with a constant probability, i.e., this probability does not depend on the identity of the vertex or its arrival time. 
The accounting in the secretary setting, however, is more complex: the probability of matching the $t$-th vertex depends on its arrival time $t$ (but {\em not} on the identity of the vertex). This still results in a constant matching probability of every vertex due to random arrival order, but leaves us with a much richer space of possible policies. Indeed, one can vary how 
aggressively we should match a vertex to its partner in the current optimal matching at different time steps $t$.
Our Algorithm~\ref{alg:vertex-arrival} is derived from the solution to the respective optimization problem. 
Interestingly, in the exploitation phase, it does not condition the probability of matching on the time $t$ (unlike our edge-arrival algorithm, see below), and it results in a tight competitive ratio of $5/12$. We remark that the optimization problem is quite subtle compared to the simple constant probability policy used in the prophet setting from~\cite{EzraFGT20}.

\paragraph{Vertex arrival: tight upper bound.}  The tightness of the $1/e$ result for the classical secretary setting is often proved in traditional algorithms course in the {\em ordinal} model (where the algorithm observes only pairwise comparisons between elements and not associated numerical values). It is known (see, e.g.,~\cite{Gnedin1994}), but is already quite non trivial, to establish the $1/e$ upper bound for the {\em game of googol} --- the {\em cardinal} version of the classical secretary problem --- where the algorithm observes a random sequence of arbitrary large numbers and wishes to maximize the probability of stopping at the maximum number among them. In general, proving upper bounds for the cardinal variants of the secretary problem is notoriously difficult. For example, only recently Correa, D\"{u}tting, Fischer, and Schewior~\cite{correa2019prophet} provided a rather complex proof that the best competitive ratio for basic prophet secretary with unknown i.i.d. priors is $1/e$. Their proof relies on an interesting application of the infinite version of Ramsey theorem.

Establishing a tight upper bound in our problem is particularly challenging. On top of being a cardinal variant of the problem, it also has a complex combinatorial structure; namely, unlike the standard setting where the algorithm should simply decide whether or not to choose an element, our algorithm observes multiple edge values and should pick a set of edges. Moreover, the objective is to maximize the expected value of the selected matching and not the probability of picking the maximum element. Despite these difficulties, we managed to construct an instance of the matching secretary problem such that: (i) any $\alpha$-competitive online algorithm for the cardinal variant can be converted into an $\alpha$-competitive algorithm in the ordinal setting, where the online algorithm can only do pairwise comparisons between edges and the objective is to maximize the probability of selecting the maximum valued edge, and (ii) no online algorithm in the ordinal setting can be better than $5/12$-competitive. Our proof of (i) is similar to~\cite{correa2019prophet} (e.g., we also use the infinite version of Ramsey theorem), but it requires several novel ideas compared to~\cite{correa2019prophet} due to the more complex combinatorial nature of our matching problem. The derivation of (ii) requires substantial work; both to prove that a certain online algorithm in the ordinal setting is optimal for the given instance, and to calculate its competitive ratio.

\paragraph{Edge arrival.} Similar to~\cite{EzraFGT20} for the matching prophet setting, we consider the matching secretary model with {\em edge} arrival. \cite{EzraFGT20} took an OCRS approach, where the idea is to control the probability of matching a realized edge $uv$, ensuring that it is a constant fraction $\alpha$ of the probability that $uv$ appears in the optimum. On the one hand, it is desired to have $\alpha$ as large as possible. On the other hand, the event that both vertices $u$ and $v$ are not yet matched upon the arrival of the edge $uv$ should have a sufficiently high probability. Following a simple union bound argument, setting $\alpha=1/3$ was sufficient for the matching prophet setting. 

We take an analogous approach for the matching secretary, namely, we control the probabilities of selecting the last arriving edge at time $t$ given that it appears in the current optimal matching. 
However, unlike the simple solution of \cite{EzraFGT20}, we cannot simply set this probability to be a constant, since the edges arriving earlier generally have a higher chance to appear in the current optimal matching. Instead, we set these control probabilities $(\alpha_t)_{t=1}^{|E|}$ to be dependent on the time $t$. We obtained a recurrent relation on $(\alpha_t)_{t=1}^{|E|}$ using a union bound argument similar to ~\cite{EzraFGT20}, which ensures that upon the arrival of an edge $uv$, both ends of the edge are not yet matched with a sufficiently high probability. Given a sequence $(\alpha_t)_{t=1}^{|E|}$, one can derive a good estimate on the competitive ratio of the corresponding online algorithm. We solved the resulting constrained optimization problem and obtained the sequence of $(\alpha_t)_{t=1}^{|E|}$ defined in Algorithm~\ref{alg:edge-arrival}. Interestingly, despite the conceptual simplicity of Algorithm~\ref{alg:edge-arrival}, it is unclear how to implement it in polynomial time\footnote{The difficulty is that we need to understand what the online algorithm would do for different subsets of the arrived set of vertices in order to estimate the probability that two given vertices $u$ and $v$ are both available upon the arrival of the edge $uv$.}. Thus, our result for the edge arrival model is information theoretic. Obtaining a poly-time algorithm with this ratio remains as an interesting open problem.

In conclusion, our results for both vertex and edge arrival models establish a close connection between secretary and prophet settings. 
Our results demonstrate that tools like OCRS that prove useful in prophet settings (e.g., \cite{EzraFGT20}) can be used to improve state of the art results for secretary settings. In particular, the tight result for matching prophets translates (with suitable adjustments) into a tight result for matching secretaries in the vertex arrival model. This is in contrast to previous general-purpose results~\cite{Dughmi20} connecting Contention Resolution Schemes and secretary problems, which suffered certain constant factor losses in the transition from one setting to another. 

\subsection{Related work}

The 1-sided secretary matching problem studied by \citet{KorulaP09} is a generalization of the matroid secretary problem on transversal matroids, which was first introduced by \citet{BabaioffIKK18}. They designed constant competitive algorithms for graphs with bounded degrees, and \citet{DimitrovP12} generalized this result to arbitrary graphs. These results affirmatively answer the famous matroid secretary conjecture by \citet{BabaioffIKK18} for transversal matroids. Whether $\Omega(1)$-competitive algorithms exist for the secretary problem on general matroids remains an intriguing question. Currently, the best known ratio is $\Omega(1/\log \log \text{rank})$ \citet{Lachish14} and \citet{FeldmanSZ18}. 

The $1/e$-competitive algorithm by \citet{KesselheimRTV13} for 1-sided secretary matching is further extended to a truthful mechanism that attains the same competitive ratio by \citet{Reiffenhauser19}. \citet{KorulaP09} also studied the secretary matching on hypergraphs and proposed $\Omega(1/d^2)$-competitive algorithms for $d$-hypergraphs\footnote{A $d$-hypergraph is a hypergraph such that all its hyperedges have size at most $d$.}. This result is improved to $\Omega(1/d)$-competitive by \citet{KesselheimRTV13}. 

Another line of work considers the secretary problem in the ordinal setting. That is, the algorithm is restricted to do pairwise comparisons between elements, without knowing the exact values of their weights. \citet{HoeferK17} designed constant competitive algorithms for several families of constraints, including matching, packing LPs and independent set with bounded local independence number. Specifically, they studied the same vertex-arrival model for general graphs as our model and designed an $\frac{e+1}{12e}$-competitive algorithm that only uses ordinal information. \citet{SotoTV18} studied the ordinal matroid secretary problem and achieved improved competitive ratios for transversal matroids, matching matroids, laminar matroids, etc.

The general vertex arrival model adopted in this paper is first introduced by \citet{WangW15} in the online matching literature, where the focus is to select maximum size matching under adversarial vertex arrivals. This model is a natural generalization of the 1-sided online bipartite matching model by \citet{KarpVV90}. \citet{WangW15} designed a $0.526$-competitive algorithm for the fractional version of the problem. \citet{GamlathKMSW19} proposed a $1/2+\Omega(1)$-competitive algorithm for the integral version of the problem. They also proved that no algorithm has a competitive ratio larger than $1/2$ in the edge arrival setting. Motivated by online ride-sharing, there has been a growing interest in online matching in the past few years and other extensions of the 1-sided online bipartite matching model have been studied, including fully online matching~\cite{HuangPTTWZ19, HuangKTWZZ20, HuangTWZ20} and edge-weighted online windowed matching~\cite{AshlagiBDJSS19}.

The matching prophet problem with edge arrivals is first studied by \citet{KleinbergW19} under the more general framework of matroid intersections. \citet{GravinW19} explicitly studied the bipartite matching setting and designed a threshold-based $\frac{1}{3}$-competitive algorithm. They also showed an upper bound of $\frac{4}{9}$. The upper bound was improved to $\frac{3}{7}$ in \cite{Tristanthesis}.
\citet{EzraFGT20} designed an improved $0.337$-competitive algorithm for general graphs.

\section{Model and Preliminaries}
\label{sec:model}
The setting is presented by a graph $G=(V,E)$, where $V$ is a set of $n$ vertices. 
Each edge $uv\in E$ has a weight $w_e\in\reals$ and the vector $w\in \reals^{|E|}$ contains the weights of all edges.
Given a subset of the vertices $T \subseteq V$, we denote by $G(T)$ the graph induced by $T$. 
Similarly, given a subset of the edges $E' \subseteq E$, we denote by $G(E')$ the graph induced by $E'$.   
We consider two arrival models, namely (i) vertex arrival, and (ii) edge arrival, where the arriving elements are the vertices and edges, respectively.
In both models, upon the arrival of an element, a matching decision should be made immediately and irrevocably, and the goal is to maximize the total weight of the resulting match. 


A matching $\mu$ is a subset of $E$, where every vertex is matched to at most a single other vertex.
We will also write $\mu(v)$ for the vertex matched with $v$ in the matching $\mu$.  
That is, if $uv \in \mu$ then $\mu(v) = u$ and $\mu(u) = v$.  
We denote by $w(\mu)\eqdef\sum_{e\in\mu}w_e$ the total weight of matching $\mu$.
In addition, given a subset of the vertices $T \subseteq V$ and a matching $\mu$, write $\mu|_T = \{ uv \in \mu | v,u \in T \}$ for the matching $\mu$ restricted to vertices in $T$.
For a weight function $w$, we write $\mu^*(w)$ for the maximum weighted matching under $w$. 

\paragraph{Vertex arrival}
Under vertex arrival model, the vertices arrive in a {\em uniformly random} order.  
We rename the vertices $v_1, \ldots, v_n$ according to their arrival order, so that $v_t$ is the vertex that arrives at time $t$.
We denote by $V_t = \{v_1, \ldots, v_t\}$ the set of vertices that arrived up to time $t$, and by $G(V_t)$ the graph induced by $V_t$.
Upon the arrival of vertex $v_t$, the weight $w_{v_t v_j}$ is revealed for all vertices $v_j \in V_{t-1}$. 
Consequently, $v_t$ can either be matched to some available vertex $v_j \in V_{t-1}$ (in which case $v_t$ and $v_j$ are marked as unavailable) or left unmatched (in which case $v_t$ remains available for future  matches).
We assume that the number of vertices $n$ is known.

Without loss of generality, we may assume that $G$ is a complete graph: we simply add 0-weight edges for the missing edges. Thus we may assume that for every $V' \subseteq V$ such that $|V'|$ is even, the maximum weighted matching of $G(V')$ matches all vertices of $G(V')$.  

\paragraph{Edge arrival}
Under edge arrival model, the edges arrive in a {\em uniformly random} order.
We rename the edges $e_1, \ldots, e_m$ according to their arrival order, so that $e_t$ is the edge that arrives at time $t$. 
We denote by $E_t = \{e_1, \ldots, e_t\}$ the set of edges that arrived up to time $t$, and by $G(E_t)$ the graph induced by $E_t$.
Upon the arrival of edge $e_t=(u,v)$, its weight $w_{e_t}$ is revealed. 
If both $u$ and $v$ are available, then $e_t$ can either be matched (in which case $u$ and $v$ are marked as unavailable) or left unmatched (in which case $u$ and $v$ remain available for future matches). 
We assume that the number of edges $m$ is known.

We assume without loss of generality that the maximum weighted matching of $G(S)$ for any $S\subseteq E$ in both vertex and edge arrival model is unique. Indeed, we can perturb the weight of every edge by adding to it a random number in $[0,\epsilon]$, for a sufficiently small $\epsilon$.

\section{Secretary Matching with Vertex Arrival}
In this section, we present an algorithm that gives a competitive ratio of $5/12$. The algorithm ignores the first $k$ vertices (exploration phase). 
Then, in every round $t$, it makes sure that the number of vertices is even. It does so by removing a random vertex $r_t$ from $V_{t-1}$ when $t$ is odd. 
The algorithm finds maximum weighted matching $\mu_t$ in the graph $G(V_t)$ ($G(V_t\setminus\{v_{r_t}\})$ if $t$ is odd). Since the matching is complete, $v_t$ must have a partner $\mu_t(v_t)$ in this matching. If this partner is available, the algorithm matches $v_t$ to it. See Algorithm~\ref{alg:vertex-arrival}.

\begin{algorithm}
	\caption{$5/12$-matching secretary for vertex arrival (for $k=\lfloor\frac{n}{2}\rfloor$)}
	\label{alg:vertex-arrival}
	\begin{algorithmic}[1]
		\STATE Let $v_1,\ldots,v_n$ be the vertices in arrival order
		\STATE $A=V$
		\quad\quad\quad\quad\quad\quad\quad\quad
		\tcp{$A$ is the set of available vertices}
		\STATE $\mu=\emptyset$  \quad\quad\quad\quad\quad\quad\quad\quad
		\tcp{$\mu$ is the returned matching}
		\FOR{$t \in \{k+1,...,n\}$} 
		\STATE Let $V_t = \{v_1,\ldots,v_t\}$
		\quad\quad\quad\quad\quad\quad\quad\quad
		\tcp{$V_t$ is the set of vertices arrived up to time $t$}
		\IF{$t$ is odd} 
		\STATE Select $r_t \in \{1,\ldots t-1\}$ uniformly at random 
		\STATE Set $V'_t=V_t \setminus \{v_{r_t}\}$ \label{line:drop}
		\quad\quad\quad\quad\quad\quad\quad\quad
		\tcp{delete a random vertex from $v_1, \ldots, v_{t-1}$}
		\ELSE
		\STATE Set $V'_t=V_t $
		\ENDIF
		\STATE Let $\mu_t$ be the maximum weighted matching in $G(V'_t)$
		\IF{$\mu_t(v_t) \in V_{t} \cap A$}
		\STATE Add $e_t\eqdef v_t\mu_t(v_t)$ to $\mu$ 
		\quad\quad\quad\quad\quad\quad\quad\quad
		\tcp{add the chosen edge to the matching}
		\STATE Remove $v_t$ and $\mu_t(v_t)$ from $A$
		\ENDIF  
		\ENDFOR
		\STATE Return matching $\mu$
	\end{algorithmic}
\end{algorithm} 

In our analysis, the event ``$\matby{u}{t}$'' refers to the event that vertex $u$ is matched either before the arrival of $v_t$ or
exactly in the round $t$.
The following theorem asserts that Algorithm~\ref{alg:vertex-arrival} achieves a competitive ratio of $\frac{5}{12}$.

\begin{theorem}
\label{th:5-over-12}
Algorithm~\ref{alg:vertex-arrival}, with $k=\lfloor \frac{n}{2} \rfloor$, has a competitive ratio of $\frac{5}{12}$ for matching secretary with vertex arrival.
\end{theorem}

The following lemma will be used in the proof of Theorem~\ref{th:5-over-12}. 

\begin{lemma}\label{lem:pr-u-matched}
For every $t \geq k$, every possible realization $\bvt$ of $V_t$ (i.e., $\bvt \subseteq V, |\bvt|=t$), and every vertex $u \in \bvt$, it holds that 
\begin{equation}
\Prx{\matby{u}{t} \mid V_t=\bvt} = \frac{2}{3} \left(  1 -   \frac{(t-3)! \cdot k!}{t! \cdot (k-3)!} \right).
\label{eq:pr-u-matched}
\end{equation} 	
\end{lemma}
\begin{proof}
To prove the lemma, we show that the probability that $u$ is matched by time $t$, conditioned on $V_t=\bvt$ can be expressed by the following recursive formula: 
\begin{equation}
\Prx{\matby{u}{t} \mid V_t=\bvt} = p(k,t),
\label{eq:matched}
\end{equation}
where 
\begin{equation}
\label{eq:recursion}
p(k,k) = 0 \quad \mbox{ and } \quad  p(k,t) = \frac{2}{t} +\frac{t-3}{t}\cdot p(k,t-1) \mbox{ for every } t \in \{k+1,\ldots,n\}.
\end{equation}
We prove \eqref{eq:matched} by induction on $t$. 
For $t= k$, $p(k,k)=0$ and~\eqref{eq:matched} holds trivially.
Consider next the case where $t>k$, and $V_t = \bvt$. 
Recall that for every set $T \subseteq V$ of even size and $u\in T$, $\mu_{T}(u)$ denotes the match of $u$ in the maximum weighted matching in $G(T)$ (recall that the maximum matching is unique and matches all vertices).
We distinguish between two cases, namely whether $t$ is even or odd; in both cases $V'_t$ is even.

\textbf{Case 1:} $t$ is even. 
In this case $V'_t=V_t$. We partition the event that $u$ is matched by time $t$, given that $V_t=\bvt$, into the following disjoint events: (i) $v_t=u$, (ii) $v_t=v$ for some $v \in \bvt \setminus \{u,\pair{\bvt}{u}\}$, 
and (iii)  $v_t=\pair{\bvt}{u}$. Each one of these events occurs with probability $\frac{1}{t}$ (in (ii), $\frac{1}{t}$ is the probability of every given $v\in \bvt \setminus \{u,\pair{\bvt}{u}\}$). We get:
\begin{eqnarray*}
	\Prx{\matby{u}{t} \mid V_t=\bvt} & = &  
	\Prx{\matby{u}{t}\mid V_t=\bvt,  v_t=u } \frac{1}{t}\\
	& + &  \sum_{v \in \bvt \setminus \{u,\pair{\bvt}{u}\}} \Prx{\matby{u }{t} \mid V_t=\bvt, v_t=v }\frac{1}{t} \\
	& + &  \Prx{\matby{u}{t}  \mid V_t=\bvt, v_t=\pair{\bvt}{u}} \frac{1}{t}.
\end{eqnarray*}
If $v_t=u$, then $u$ is matched by time $t$, iff $\pair{\bvt}{u}$ is unmatched before $u$'s arrival, which happens with probability $1-p(k,t-1)$ by induction hypothesis for $V_{t-1}=\bvt\setminus\{u\}$. If $v_t$ is neither $u$ nor $\pair{\bvt}{u}$, then $u$ is matched by time $t$ iff it is matched by time $t-1$. Finally, if $v_t=\pair{\bvt}{u}$, then $u$ is always matched by $t$, since if it is unmatched before time $t$, it will be matched to $\pair{\bvt}{u}$ upon arrival of $v_t=\pair{\bvt}{u}$. Putting it all together we get:
\begin{eqnarray*}
\Prx{\matby{u}{t} \mid V_t=\bvt}  & = &  
(1-p(k,t-1))\frac{1}{t}  +(t-2) p(k,t-1)  \frac{1}{t} + \frac{1}{t}\\
& = & \frac{2}{t} + \frac{t-3}{t}\cdot p(k,t-1)\\
& \stackrel{\eqref{eq:recursion}}{=} & p(k,t).
\end{eqnarray*}

\textbf{Case 2:} $t$ is odd. Let $r_t$ be the index of the random vertex that is dropped in line~\ref{line:drop} of the algorithm. Then, $V'_t=\bvt \setminus \{v_{r_t}\}$. 
We partition the event that $u$ is matched by time $t$, given that $V_t=\bvt$, into the following disjoint events: 
(i) $u=v_t$, (ii) $u=v_{r_t}$, and (iii) $u\neq v_t, v_{r_t}$.  
Each of the events (i) and (ii) occurs with probability $\frac{1}{t}$; event (iii) occurs with probability $\frac{t-2}{t}$.
We get:
\begin{eqnarray*}
	\Prx{\matby{u}{t}\mid V_t=\bvt} & = &  
	\Prx{\matby{u}{t} \mid V_t=\bvt,  u=v_t } \frac{1}{t}\\
	& + & \Prx{\matby{u}{t} \mid V_t=\bvt, u=v_{r_t}} \frac{1}{t}\\
	& + & \Prx{\matby{u}{t}\mid V_t=\bvt, u\neq v_t,v_{r_t} }\frac{t-2}{t}.
\end{eqnarray*}
If $u=v_t$, then $u$ is matched iff its match is available in round $t$, which happens with probability $1-p(k,t-1)$, by induction.
If $u=v_{r_t}$, then $u$ is matched by time $t$ iff it is matched by time $t-1$, which happens with probability $p(k,t-1)$ by the induction hypothesis for $V_{t-1}=\bvt\setminus\{v_t\}$.
If $u \neq v_t,v_{r_t}$, then $v_t,v_{r_t}$ are uniformly distributed among the pairs of vertices in $\bvt \setminus\{u\}$. 
To calculate the probability that $u$ is matched by time $t$  we separate the latter case into  two disjoint events: (i) $\mu_t(v_t) = u$, in which case $u$ is matched with probability 1; and (ii) $\mu_t(v_t) \neq u$, in which case $u$ is matched only if it was matched before time $t$, which is  $p(k,t-1)$ by induction. Thus, 
\begin{eqnarray}
\Prx{\matby{u}{t}\mid V_t=\bvt, u\neq v_t,v_{r_t} } & =&  \Prx{\mu_t(v_t) = u \mid V_t=\bvt, u\neq v_t,v_{r_t} } \cdot 1 \nonumber \\ & +& \Prx{\mu_t(v_t) \neq u \mid V_t=\bvt, u\neq v_t,v_{r_t} } \cdot p(k,t-1) \nonumber\\ & =&   \frac{1}{t-2}\cdot 1 + \frac{t-3}{t-2}p(k,t-1). \nonumber
\end{eqnarray}
Putting it all together we get:
\begin{multline*}
	\Prx{\matby{u}{t}\mid V_t=\bvt}  =   
(1-p(k,t-1)) \frac{1}{t}+  p(k,t-1) \frac{1}{t}  \\
+\frac{t-2}{t}\cdot \left(\frac{1}{t-2} + \frac{t-3}{t-2}\cdot p(k,t-1)\right) 
 =
\frac{2}{t} + \frac{t-3}{t} p(k,t-1)
 =  p(k,t). \label{eq:pkt}
\end{multline*}
This concludes the proof of Equation~\eqref{eq:matched}.

It remains to solve the recursion. We prove by induction that
\begin{equation}
p(k,t)=\frac{2}{3} \left( 1 - \frac{(t-3)! \cdot k!}{t! \cdot (k-3)!}\right).
\label{eq:pkt2}
\end{equation}
For $t=k$ this holds trivially, since $p(k,k)=0$.
For $t>k$, suppose \eqref{eq:pkt2} holds for $t-1$; then,
\begin{eqnarray*}
p(k,t) & \stackrel{\eqref{eq:recursion}}{=} & \frac{2}{t} + \frac{t-3}{t}\cdot p(k,t-1)\nonumber \\
& \stackrel{\eqref{eq:pkt2}}{=} & \frac{2}{t} + \frac{t-3}{t}\cdot \frac{2}{3} \left( 1 - \frac{(t-4)! \cdot k!}{(t-1)! \cdot (k-3)!}\right) \\
& = & \frac{2}{3} \left(  \frac{3}{t} + \frac{t-3}{t}\cdot \left( 1 - \frac{(t-4)! \cdot k!}{(t-1)! \cdot (k-3)!}\right) \right) \\ 
& = & \frac{2}{3} \left(  1 -   \frac{(t-3)! \cdot k!}{t! \cdot (k-3)!} \right),\\ 
\end{eqnarray*}
where the second equality holds by the induction assumption.
This concludes the proof of Lemma~\ref{lem:pr-u-matched}.
\end{proof}

With Lemma~\ref{lem:pr-u-matched} in hand, we are ready to prove Theorem~\ref{th:5-over-12}.

\begin{proof}
Given some $t \geq k+1$, let $\mu_t$ denote the maximum weighted matching in $G(V'_t)$, and let $\mu^*$ be the maximum weighted matching in $G$. 
We give a lower bound on the expected weight of the edge $e_t$. 
Since $e_t$ is a random edge chosen  uniformly at random among the $\lfloor t/2\rfloor$ edges in $\mu_t$, for every possible realization $\bvt'$ of $V'_t$ (note that $\bvt'$ is of size $t$ if $t$ is even, and of size $t-1$ if $t$ is odd), it holds that
$$
\Ex{\wet} = \frac{\Ex{\wmt}}{\lfloor t/2\rfloor}.
$$
Recall that $\mu^*|_{\bvt'} = \{ ij \in \mu^* | i,j \in \bvt' \}$ for the matching $\mu^*$ restricted to vertices in $\bvt'$.
Since $\mu_t$ is the maximum matching selected for the set of vertices $V'_t$, it holds that $\Ex{w(\mu_t)} \geq \Ex{w(\mu^*|_{V'_t})}$.
Consider next $\Ex{w(\mu^*|_{V'_t})}$. The number of edges in $V'_t$ is ${2\lfloor \frac{t}{2}\rfloor \choose 2}$. Since all edges are symmetric (by random arrival) and the total number of edges is ${n \choose 2}$, it holds that $\Ex{w(\mu^*|_{V'_t})} = \Ex{w(\mu^*)}{2\lfloor \frac{t}{2}\rfloor \choose 2}  / {n \choose 2}$. 
We get: 
\begin{eqnarray*}
	\Ex{\wet  } & \geq &  \frac{1}{\lfloor t/2\rfloor} \Ex{w(\mu^*|_{\bvt'})}\\
	& = & \frac{1}{\lfloor t/2\rfloor} \Ex{w(\mu^*)} {2\lfloor \frac{t}{2}\rfloor \choose 2}  / {n \choose 2} \\
	& = & \frac{4\cdot \lfloor t/2\rfloor - 2 }{n \cdot (n-1) } \Ex{w(\mu^*)}.
\label{eq:part-opt}
\end{eqnarray*}


Putting it all together, writing $\mu$ for the matching obtained by Algorithm~\ref{alg:vertex-arrival}, and $\mu^*$ for the optimal matching, we get:
$$
\frac{\Ex{w(\mu)}}{\Ex{w(\mu^*)}} =  \frac{1}{\Ex{w(\mu^*)}}\sum_{t=k+1}^n \Prx{\mu_t(v_t) \in V_{t} \cap A_t} \Ex{\wet},
$$
where $A_t$ denotes the set $A$ in the beginning of iteration $t$.

Substituting $\Prx{\mu_t(v_t) \in V_{t} \cap A} = 1-p(k,t-1)$, and applying Equation~\eqref{eq:part-opt}, we get
\begin{eqnarray}
\frac{\Ex{w(\mu)}}{\Ex{w(\mu^*)}} & = & \frac{1}{\Ex{w(\mu^*)}}\sum_{t=k+1}^n (1-p(k,t-1))\cdot\frac{4\cdot \lfloor t/2\rfloor -2 }{n \cdot (n-1) } \cdot {\Ex{w(\mu^*)}}  \nonumber \\
& \stackrel{\eqref{eq:pkt2}}{=} & \sum_{t=k+1}^n \left(1- \frac{2}{3} \left(  1 -   \frac{(t-4)! \cdot k!}{(t-1)! \cdot (k-3)!} \right)\right)\cdot\frac{4\cdot \lfloor t/2\rfloor -2 }{n \cdot (n-1) }  \nonumber \\
& \geq & \sum_{t=k+1}^n \left(\frac{1}{3}+  \frac{2\cdot (t-4)! \cdot k!}{3\cdot (t-1)! \cdot (k-3)!} \right)\cdot\frac{2t -4 }{n \cdot (n-1) }  \nonumber \\
& \geq & \sum_{t=k+1}^n \left(\frac{1}{3}+  \frac{2 (k-2)^3}{3\cdot t^3} \right)\cdot\frac{2t -4 }{n^2 }  \nonumber \\
& \geq & \frac{1}{n^2}\int_{k}^n \left(\frac{1}{3}+  \frac{2 (k-2)^3}{3\cdot (t+1)^3} \right)\cdot(2t -4)  \dd\,t\nonumber \\
& = & \frac{1}{n^2}\cdot \frac{1}{3}\left( -\frac{4(k-2)^3}{t+1}+ \frac{6(k-2)^3}{(t+1)^2} +(t+1)^2-6t\right) \Biggr|_k^n \nonumber \\
&=& -\frac{4k^3}{3n^3} +\frac{4k^2}{3n^2} +\frac{1}{3} - \frac{k^2}{3n^2} -o(1) \nonumber \\
& =& \frac{1}{3} +\frac{k^2}{n^2}-\frac{4k^3}{3n^3} -o(1), \nonumber
\end{eqnarray} 
where to get the first inequality, we used the bound $\lfloor t/2\rfloor \ge t/2-1/2$; to get the second inequality, we applied basic algebraic transformations to simplify each term under the summation; to get the third inequality, we estimated the integral $\int_{t=k}^{t=n}$ as a Riemann sum with the subdivision into equal intervals of length $1$ and used a simple upper bound on the function's value in each subdivision interval; in the last two equalities we collected all low-order terms in $o(1)$ notation (that vanishes when $n\to\infty$ and $k=\Theta(n)$).

The last expression attains its maximum at $k=\frac{n}{2}$, achieving a competitive ratio of $\frac{5}{12}-o(1)$.
\end{proof}

{\bf Remark:} Note that for every graph with three vertices, Algorithm~\ref{alg:vertex-arrival} always matches $v_3$ to the unique vertex in $\{v_1,v_2\} \setminus \{v_r\}$. Therefore, it always achieves a competitive ratio of $\frac{1}{3}$. However,
the modified algorithm that adds $m>>n$ auxiliary vertices that are connected to all vertices with zero weight edges, and then applies Algorithm~\ref{alg:vertex-arrival} (where the  auxiliary vertices are added at random times) gives a competitive ratio that approaches $\frac{5}{12}$ as $m$ goes to infinity for all $n$.

{\bf Remark:} Algorithm~\ref{alg:vertex-arrival} can be modified to a $5/24$-competitive algorithm in the ordinal setting. Observe that the only step Algorithm~\ref{alg:vertex-arrival} uses weights of the edges is for constructing a maximum weighted matching $\mu_t$ in line 12. We modify this step by constructing $\mu_t$ greedily. That is, we keep adding the largest edge to the matching until all vertices are matched. This procedure can be implemented using pairwise comparisons of edges. In this way, $\mu_t$ is a $2$-approximation to the maximum matching in $G(V_t')$ and we suffer an extra factor of $2$. The rest of our analysis remains intact. The resulting $5/24\approx 0.208$-competitive ratio improves upon the $\frac{e+1}{12e} \approx 0.114$ competitive ratio by \citet{HoeferK17}.


\section{Upper Bound for Vertex Arrival}
In this section we establish the following theorem showing that the competitive ratio of $\frac{5}{12}$ is tight.
\begin{theorem}
	\label{thm:upper-bound-vertex}
	No online algorithm has a competitive ratio better than $\frac{5}{12}$ for matching secretary with vertex arrival.
\end{theorem}

The proof of Theorem~\ref{thm:upper-bound-vertex} is composed of the following components:
\begin{enumerate}
	\item We introduce an {\em ordinal} variant of the matching secretary problem (Section~\ref{sec:ordinal}).
	\item We reduce our matching secretary problem to the ordinal variant (Section~\ref{sec:reduction}).
	\item We establish an upper bound of $5/12$ with respect to the ordinal variant (Section~\ref{sec:upper-bound-ordinal}).
\end{enumerate}

\subsection{The Ordinal Variant}
\label{sec:ordinal}

The {\em ordinal} variant of the matching problem is the following:
	\begin{itemize}
		\item A set of $n$ vertices are ranked, according to an unknown ranking, from 1st to $n$th. The 1st and 2nd vertices are referred to as the top two vertices.
		\item The vertices arrive sequentially, in a random order; let $v_1, \ldots, v_n$ denote the vertices in their arrival order.
		\item Upon the arrival of vertex $v_t$, the algorithm observes the relative rank of $v_t$ among $v_1, \ldots, v_t$ (its rank is in $\{1, \ldots, t\}$), and must decide immediately and irrevocably whether to match it to an earlier unmatched vertex.
		\item The objective is to maximize the probability of matching together the top two vertices. 
	\end{itemize}


In the remainder of this section, we refer to our original matching secretary setting as the \emph{cardinal} setting, and to this variant as the \emph{ordinal} setting. 
Note that the two settings differ both in (i) the assumption about what is observable (a vertex's weight in the cardinal setting versus its relative rank in the ordinal setting), and in (ii) the objective function (maximize the expected total weight in the cardinal setting versus maximize the probability to match the top two vertices in the ordinal setting). 
The reduction will go through a third variant, which we refer to as the {\em intermediate} setting, which shares properties with both variants, as will be explained in Section~\ref{sec:reduction}.

An algorithm in the ordinal and intermediate settings is said to be $\alpha$-competitive if it matches together the top two vertices with probability at least $\alpha$.

%

\paragraph{Ordinal vs. Cardinal Classical Secretary}
It is worthwhile to mention that the classical secretary problem has two variants as well: (i) the ordinal secretary problem, where the algorithm observes the relative rank of the arriving element, and aims to maximize the probability of selecting the best element, and (ii) the cardinal secretary problem, where each element is associated with a value, which is observed upon arrival, and the algorithm aims to maximize the expectd value of the selected element.

It is straightforward to see that any algorithm in the ordinal setting preserves its competitive ratio when applied to the cardinal setting. 
On the other direction, a folklore result says that the cardinal setting is not easier than the ordinal one (recall that the best competitive ratios in both settings is $\frac{1}{e}$). The upper bound for the cardinal setting is nicely explained in the recent work of \citet{CorreaDFS19}.

Our ordinal and cardinal variants for matching can be viewed as analogs of the cardinal and ordinal settings in the classic secretary problem. 
However, the matching setting is more involved in its combinatorial structure, and the multiple decisions that should be made. 
Indeed, it is not clear to us whether it is possible to adapt an algorithm for the ordinal setting to the cardinal matching setting (like in the classical secretary problem).
Nevertheless, we establish a reduction from the cardinal setting to the ordinal setting. 




\subsection{Reduction: From Cardinal to Ordinal (through Intermediate)}
\label{sec:reduction}



We shall focus on the following family of instances. An instance is described by a complete graph $G$ on $n$ vertices. Each vertex $v$ of the graph is associated with a value $\val_v \in \naturals$, where $\naturals$ is the set of positive integers. The weight of each edge $uv$ is determined by the values of the two endpoints:  $\weight_{uv} = n^{3(i+j)}$ where $i=\val_u, j=\val_v$ are the values of its two endpoints.
We assume that $\val_u \neq \val_v$ for every two distinct vertices $u,v$. Thus, every instance can be specified by a set of values $\Val \subset \naturals$ of size $n=|\Val|$ on the graph vertices. Let $\subsets{n}$ denote the set of all such subsets $\Val$ and in general $\subsets[T]{i}$ denote the set of all subsets $\Val\subset T$ with $|\Val|=i$.
Without loss of generality, we assume that the algorithm observes the vertex values directly, rather than the edge weights. 
When clear in the context we refer to $ij$ as the edge between vertices with values $i$ and $j$, and denote the weight of this edge by $\weight_{ij}$.

The reduction from the cardinal setting to the ordinal setting proceeds in three main steps; an overview is given in Figure~\ref{fig:proof-overview}.
In what follows, we give details for each step. 

\begin{figure}[H]
	\centering
	\includegraphics[width=0.7\textwidth]{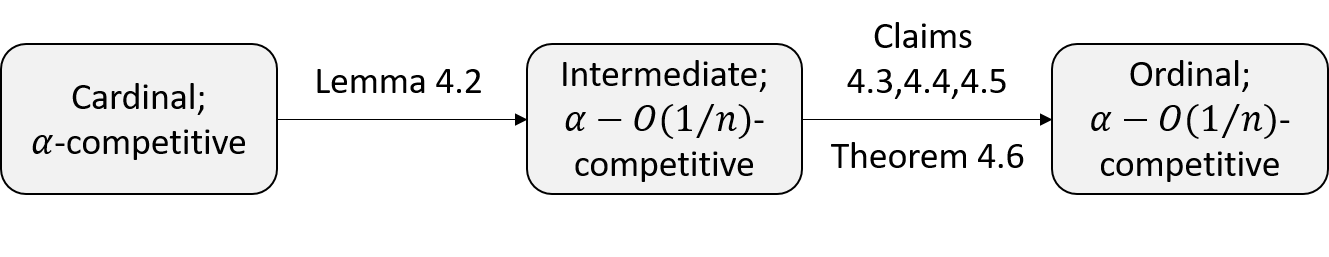}
	\caption{An overview of the proof steps. 
		\label{fig:proof-overview}}
\end{figure}

\paragraph{Step 1: Introducing the intermediate variant, and reducing the cardinal variant to the intermediate variant.}
The {\em intermediate} variant of the matching problem bridges between the cardinal and ordinal variants.
In the intermediate variant, every vertex $v$ is associated with value $\val_v$, which is observed upon $v$'s arrival (as in the cardinal settings), but the objective is to maximize the probability to match together the two vertices with the highest values (similar to the ordinal setting). 
The following lemma reduces the cardinal variant of the problem to the intermediate variant.

\begin{lemma}
\label{lem:obj_ordinal}
Let  $\Val\in \subsets{n}$. If $\alg(\Val) \ge \alpha \cdot \opt(\Val)$ in the cardinal setting, then $\alg$ matches the top two vertices in $\Val$ with probability $\alpha - O(\frac{1}{n})$.
\end{lemma}
\begin{proof}
Let $i_1, i_2$ be the two largest values of $\Val$. 
Consider the performance of the algorithm.
\begin{align}
\alg(\Val) = & \sum_{ \{i,j\} \in {\Val \choose 2}} \Prx{ij \text{ is selected}} \cdot \val_{i j} \nonumber \\
\le & \Prx{i_1i_2 \text{ is selected}} \cdot \val_{i_1 i_2} + \sum_{ \{i,j\} \in {\Val \choose 2}-\{i_1,i_2\}} \val_{i j} \nonumber \\
\le & \Prx{i_1i_2 \text{ is selected}} \cdot \val_{i_1 i_2} + \left({n \choose 2} - 1\right) \cdot \frac{\val_{i_1 i_2}}{n^3} \nonumber \\
= & \left( \Prx{i_1i_2 \text{ is selected}} + O\left(\frac{1}{n}\right) \right) \cdot \val_{i_1 i_2}, \label{eq:exp2prob}
\end{align}
where the second inequality follows from the fact that for every $i,j$, $\weight_{ij} = n^{3(i+j)} \le n^{3(i_1+i_2)-3} = \frac{\val_{i_1 i_2}}{n^3}$.

Finally, we have that $\alg(\Val) \ge \alpha \cdot \opt(\Val) \ge \alpha \cdot \val_{i_1 i_2}$. Combining this with equation~\eqref{eq:exp2prob} concludes the proof.
\end{proof}

Lemma~\ref{lem:obj_ordinal} reduces the cardinal setting to the intermediate setting. 
As such, it allows to change the objective in the cardinal setting to the objective in the ordinal setting. 
Note, however, that changing only the objective is not sufficient, as the online algorithm in the ordinal setting does not observe the values of the vertices, rather it observes only the relative ranking among them. 
In the following steps we reduce the intermediate variant to the ordinal variant.

\paragraph{Step 2: Simplifying the description of the algorithm for the intermediate variant.}
We now show that any $\alpha$-competitive online algorithm for the intermediate variant admits a simplified description. 
Claim~\ref{cl:only_top_two} shows that the algorithm's decision at each step $t$ is simply a binary decision, and 
Claim~\ref{cl:history_independent} shows that this decision is history independent.

\begin{claim}
\label{cl:only_top_two}
Let $\alg$ be any $\alpha$-competitive online algorithm in the intermediate setting. 
Then, there is an $\alpha$-competitive algorithm $\alg'$ that in every step $t$, may (but not necessarily) only match the top two vertices up to step $t$.
\end{claim}
\begin{proof}
Suppose that at some step $t$, $\alg$ matches two vertices $u$ and $v$ which are not the top two vertices up to step $t$. Consider a modified algorithm $\alg'$ that in such cases does not match any vertices, and proceeds as $\alg$ thereafter as if $u$ and $v$ were matched. Then $\alg'(\Val)=\alg(\Val)$ for any set of values $\Val$.
\end{proof}

Thus, the algorithm's decision at time $t$ is just a choice between (at most) two options: match the top two vertices so far, or not match them. However, such decision may still be quite complex: it may be randomized, it may depend on the values of the available vertices, on the arrival order, or on the previous choices of the algorithm. Our next claim simplifies the algorithm further. It shows that an optimal online algorithm is \emph{history-independent}, i.e., the decision depends only on the vertex values and on whether matching the top two vertices is possible. 

\begin{claim}
\label{cl:history_independent}
Given set of vertex values $\Val_t$ at time $t$, 
we may assume without loss of generality that the algorithm's binary decision depends only on
\begin{enumerate}
\item whether the top two vertices in $\Val_t$ are not matched yet; and
\item whether the last arriving vertex is one of the top two vertices in $\Val_t$. 
\footnote{Interestingly, it does not matter whether the last arriving vertex is the highest or second-highest, only whether it is one of the top two vertices.}
\end{enumerate}
\end{claim}

\begin{proof}
To be formal, let us denote 
the arrival order of the vertices in $\Val_t$ and prior decisions of the algorithm as a history $H(t,\Val_t)$. Consider any optimal online algorithm $\alg$. We are going to modify $\alg$ so that it becomes \emph{history-independent} at time $t$ with the same performance guarantee.

Fix the set of values $\Val_t$. We consider the set $\dhistory(t, \Val_t)$ of all histories $H(t,\Val_t)$ for which the online algorithm has an option to match the two vertices with the highest values in $\Val_t$, i.e., situations that satisfy both conditions 1 and 2 from the statement of the Claim~\ref{cl:history_independent}. Let us define an online algorithm $\alg'$ that for every observed history $H(t,\Val_t)\in\dhistory(t,\Val_t)$ forgets the real history and instead generates a contrafactual history $H'(t,\Val_t)$ drawn independently at random from $\dhistory(t, \Val_t)$, i.e., we generate a history $H'\in \dhistory(t, \Val_t)$ with the probability proportional to $\prob{H(t, \Val_t) = H' ~~|~~ \Val_t}$. We keep $\alg'$ unchanged for any other set of values $\Val_t$ or history $H\notin\dhistory(t, \Val_t)$.
Starting from time $t$ the algorithm $\alg'$ assumes that the real history was $H'$ (not $H$) and makes all its future decisions according to $H'$. Note that by Claim~\ref{cl:only_top_two} $\alg'$ would be a feasible algorithm. The performance of $\alg'(H ~|~\Val_t)=\alg(H' ~|~\Val_t)$, since all previous decisions of $\alg$ before time $t$ did not affect the two highest value vertices in $\Val_t$ and, therefore, did not affect the two highest value vertices in $\Val$. Furthermore, the distribution of the contrafactual histories $H'$ coincides by construction with the distribution of the actual histories $H\in\dhistory(t, \Val_t)$. Thus the expected performance of $\alg'$ is the same as $\alg$. Note that $\alg'$ is history independent for the given set of values $\Val_t$.

To conclude the proof, we need to apply the above transformation of an optimal algorithm $\alg$ for every possible set of values $\Val_t$ and every time $t\in [n]$. That can be easily done using, e.g., a backward induction on $n$.
\end{proof}

By Claims~\ref{cl:only_top_two} and~\ref{cl:history_independent} we may restrict our attention to history-independent, binary-decision online algorithms. Any such algorithm can be characterized by a collection of $n$ set functions  $f_i : \subsets{i} \to [0,1]$ for each $i\in[n]$, where for every $\Val_i\in\subsets{i}$, $f_i(\Val_i)$ denotes the probability that the algorithm matches the top two vertices in $\Val_i$ given that it is {\em possible}; i.e., that (i) one of these vertices arrives at time $i$, and (ii) both of them are unmatched. 

\paragraph{Step 3: Reducing the intermediate variant to the ordinal variant.}
To complete the reduction from the intermediate setting to the ordinal setting, we use a similar approach to the upper bound proof from \citet{CorreaDFS19}. We fix an $\alpha$-competitive online algorithm $\alg$ for the intermediate variant, which is represented by a set of functions $f_i$ for every $i\in [n]$.

Our goal is to find an infinite (or sufficiently large) subset $T \subseteq \naturals$ of values on which the algorithm's decisions do not depend on the actual values of the vertices, i.e., every function $f_i$ is a constant on each $\Val_i\in\subsets[T]{i}$. In this case $\alg$ does not use any information about the actual values for every set of values $\Val\subset T$ with $|\Val|\le n$. That is, we can use $\alg$ in the ordinal setting.
We achieve this goal within an arbitrary small additive error $\eps$, i.e., $f_i(\Val_i)=c_i \pm O(\eps)$ for every $i\in [n]$, and every $\Val_i\in\subsets[T]{i}$. 

\begin{claim}
\label{cl:ramsey}
For any collection of set functions $f_i:\subsets{i}\to [0,1]$, $i\in[n]$ and any $\eps>0$ there is an infinite set $T\subset\naturals$ and constants $c_1,\ldots,c_n\in[0,1]$, s.t. $f_i(\Val_i)=c_i+O(\eps)$ for all $\Val_i\in\subsets[T]{i}$, $i\in[n]$.
\end{claim}
\begin{proof}
The proof uses the infinite version of Ramsey theorem. We find such a set $T$ iteratively for $i\in[k]$, starting with $k=1$ and up to $k=n$. We proceed by induction on $k$. The base of the induction is the case of $k=0$, which holds trivially. 
Suppose, by the induction hypothesis, that we have an infinite set $T_{k}\subset \naturals$ and a set of constants $c_1,\ldots,c_k$ such that
$f_i(\Val_i)=c_i+O(\eps)$ for all $\Val_i\in\subsets[T_k]{i}$ and $i\in[k]$. Our goal is to find an infinite subset $T_{k+1}\subset T_k$
that satisfies the desired condition for $f_{k+1}$. Consider a complete hyper-graph on the set of vertices $T_k$ with hyper-edges of size $k+1$. Each edge $\Val\subset T_k$,  $|\Val|=k+1$, is colored in one of $1/\eps$ colors: assign color $\left\lfloor f_{k+1}(\Val)/ \eps\right\rfloor$ to the edge $\Val$. By the infinite version of Ramsey theorem \cite{ramsey2009problem} 
, this hyper-graph admits an infinite monochromatic clique. Let the color of such a clique be $C_{k+1}$, and let the set of vertices in the clique be $T_{k+1}\subset T_k$. Set the constant $c_{k+1}=\eps\cdot C_{k+1}$. Then  $c_{k+1}\le f_{k+1}(\Val) < c_{k+1}+\eps$ for any $\Val\in\subsets[T_{k+1}]{k+1}$, i.e., $f_{k+1}(\Val)=c_{k+1}+O(\eps)$ for any $\Val\subset T_{k+1}$, $|\Val|=k+1$.
\end{proof}

With this, we can conclude the reduction from the cardinal setting to the ordinal setting. 
As a bonus, we also reveal useful properties of an optimal algorithm in the ordinal setting.

\begin{theorem}
\label{th:reduction_ordinal}
Let $\alg$ be an $\alpha$-competitive online algorithm for the cardinal matching setting. 
Then, there is a $\alpha-O(\frac{1}{n})$-competitive online randomized algorithm $\alg^o$ for the ordinal setting. 
Moreover, in every time $t$, $\alg^o$ makes a binary decision on whether to match the two top vertices up to time $t$, and this decision may only depend on the time $t$ and on whether this matching is possible.
\end{theorem}
\begin{proof}
We use Lemma~\ref{lem:obj_ordinal} to convert $\alg$ into an $\alpha-O(\frac{1}{n})$-competitive algorithm $\alg^{\text{int}}$ for the intermediate setting. By Claims~\ref{cl:only_top_two} and~\ref{cl:history_independent} we can assume without loss of generality that $\alg^{\text{int}}$ makes decisions that depend only on the set of values $\Val_t$ at time $t$ and the possibility of matching the top two vertices in $\Val_t$. Finally, we can choose $\eps$ in Claim~\ref{cl:ramsey} to be $\eps=O(\frac{1}{n^2})$ and find set $T\subset \naturals$ such that $f_i(\Val_i)=c_i+O(\eps)$ for every $\Val_i\in\subsets[T]{i}$ and $i\in[n]$. 
Define the algorithm $\alg^o$ for the ordinal setting as follows: 
Upon the arrival of the $i$th vertex, match the two top vertices up to time $i$ with probability $c_i$, if this is possible. 
By the union bound,for every set $\Val\in\subsets[T]{n}$, $\alg^o(\Val)=\alg^{\text{int}}(\Val)+n\cdot O(\eps)=\alg^{\text{int}}(\Val)+O(\frac{1}{n})$ for every set. 
Thus, $\alg^o$ is an $\alpha-O(\frac{1}{n})$-competitive online algorithm for the ordinal setting, whose binary decisions depend only on the time $t$ and the possibility of matching the two top vertices in $\Val_t$ at time $t$.
This concludes the proof
\end{proof}

\subsection{An Upper Bound for the Ordinal Setting}
\label{sec:upper-bound-ordinal}


In this section we study the ordinal setting. 
Based on the analysis of Section~\ref{sec:reduction}, we restrict ourselves, without loss of generality, to algorithms that decide at each step whether to match the top two vertices so far, and the decision depends only on the time $t$ and whether this matching is possible.
Such an algorithm can be fully characterized by a vector $\vec{c} \in [0,1]^n$, where $c_i$ is the probability that the algorithm matches the top two vertices so far at step $i$, given that it is possible to match them. 
	
Our main theorem in this section is an upper bound of $5/12$ on the competitive ratio of any algorithm in the ordinal setting.
\begin{theorem}
	\label{thm:ordinal-5-12}
	In the oridinal setting, for any $\vec{c} \in [0,1]^n$, the corresponding algorithm matches the top two vertices with probability at most $\frac{5}{12} + O(\frac{1}{n})$.
\end{theorem}
\begin{proof}
Let $v_i$ denote the vertex that arrives at time $i$, and let $\fst_i,\snd_i$ denote the respective top and second-top vertices among $v_1, \ldots, v_i$. 
Let $O_i$ be the event that $\fst_i$ remains unmatched by the end of step $i$, and let $p_i = \Prx{O_i}$, where the randomness is taken over the arrival order of the first $i$ vertices. 
Clearly, $p_1=1$. In what follows, ``$\matat{e}{t}$" denotes the event that edge $e$ is matched exactly at time $t$.
For $i>1$, we can express $p_i$ by the following recursive formula:  
\begin{align}
p_{i} = & \Prx{O_{i} \wedge (v_i = \fst_i) } + \Prx{O_{i} \wedge (v_i = \snd_i)} + \Prx{O_{i} \wedge (v_i \ne \fst_i,\snd_i)} \nonumber \\
= & \frac{1}{i} \cdot \Prx{O_{i} \big| v_i=\fst_i} + \frac{1}{i} \cdot \Prx{O_{i} \big| v_i = \snd_i} + \frac{i-2}{i} \cdot \Prx{O_{i} \big| v_i \ne \fst_i,\snd_i} \nonumber \\
= & \frac{1}{i} \cdot \Prx{\widebar{O_{i-1}} \vee \overline{\matat{v_i \fst_{i-1}}{i}} ~\big|~ v_i=\fst_i} + \frac{1}{i} \cdot \Prx{O_{i-1} \wedge \overline{\matat{v_i \fst_{i-1}}{i}}~ \big|~ v_i=\snd_i} \nonumber \\
& + \frac{i-2}{i} \cdot \Prx{O_{i-1}} \nonumber \\
= & \frac{1}{i} \cdot (1-p_{i-1}c_i) + \frac{1}{i} \cdot p_{i-1} (1-c_i) + \frac{i-2}{i} \cdot p_{i-1} \nonumber \\
= & \frac{1}{i} \cdot (1+(i-1)p_{i-1} - 2p_{i-1}c_i). \label{eq:pi}
\end{align}
The first equality follows by considering three disjoint cases for vertex $v_i$; the second equality holds since the distribution over the arrival order of $\{v_j\}_{j\in[i]}$ is uniform; to get the third equality we further consider cases where $\fst_i$ will be matched: (a) if $v_i$ is the top vertex so far, then it stays unmatched if either the $\snd_i=\fst_{i-1}$ is already matched, \emph{or} we don't match $v_i$ to $\fst_{i-1}$, (b) if $v_i$ is the second top vertex so far, then $\fst_i=\fst_{i-1}$ stays unmatched if it is unmatched at the step $i-1$ \emph{and} it is not matched to $v_i$ at step $i$, (c) if $v_i$ is ranked lower than $\fst_i,\snd_i$, then $\fst_i=\fst_{i-1}$ stays unmatched if it is unmatched at step $i-1$; to obtain the fourth equality, we use the fact that the probability $c_i$ of matching $\fst_i$ and $\snd_i$ (if it is possible) depends only on the time step $i$.

We next claim that the expected performance of the algorithm can be written as
\begin{equation}
	\label{eq:obj}
	\alg(\vec{c}) = \frac{1}{\binom{n}{2}} \sum_{i=2}^{n} (i-1) \cdot p_{i-1} \cdot c_i.
\end{equation}
To see this, note that
\begin{align}
\alg(\vec{c}) = & \sum_{i=2}^{n} \Prx{\{\fst_i,\snd_i \} = \{\fst_n,\snd_n\} \wedge v_i \in \{\fst_i, \snd_i\} \wedge O_{i-1} \wedge \matat{\fst_i \snd_i}{i} } \nonumber \\
= & \sum_{i=2}^n \Prx{\{\fst_i,\snd_i\} = \{\fst_n,\snd_n\}} \cdot \frac{2}{i} \cdot p_{i-1} \cdot c_i \nonumber \\
= & \sum_{i=2}^{n} \frac{\binom{n-2}{i-2}}{\binom{n}{i}} \cdot \frac{2}{i} \cdot p_{i-1} \cdot c_i \nonumber \\
= & \frac{1}{\binom{n}{2}} \sum_{i=2}^{n} (i-1) \cdot p_{i-1} \cdot c_i, \nonumber
\end{align}
where the first equality follows from the law of total probability; the second equality follows from the following facts:
1) conditioned on $\{\fst_i,\snd_{i}\} = \{\fst_n,\snd_n\}$, the arrival order of $v_1, \ldots, v_i$ is chosen uniformly at random,
2) the events $v_i \in \{\fst_i,\snd_i\}$ and $O_{i-1}$ are independent,
3) our algorithm matches $\fst_i\snd_{i}$ with probability $c_i$ when possible, i.e., when $v_i \in \{\fst_i, \snd_i\}$ and $\fst_{i-1}$ is unmatched.

For notation simplicity, let $q_i = i p_i$ for $i \in [n]$.
By equations~\eqref{eq:pi} and \eqref{eq:obj}, in order to find the optima$\vec{c} \in [0,1]^n$, it suffices to find the maximum of the following function\footnote{$f$ does not depend on $c_1$.} on $\vec{c} \in [0,1]^n$:
\begin{equation}
\label{eq:qi-recursive}
f(\vec{c},\vec{q}) \eqdef \sum_{i=2}^n q_{i-1} c_i
\quad\quad\text{where } q_i = 1+ \left(1-\frac{2c_i}{i-1}\right) \cdot q_{i-1}, \text{ and }q_1 = 1
\end{equation}
We calculate the derivative of $f(\vec{c},\vec{q}(\vec{c}))$ over $c_i$ (note that $\vec{q}$ also depends on $\vec{c}$).
\begin{align*}
\frac{df}{dc_i} = & q_{i-1} + \sum_{j=i+1}^{n} c_j \cdot \frac{dq_{j-1}}{d c_i} & \mbox{($q_j$ does not depend on $c_i$ for $j < i$)}\\
= & q_{i-1} + \sum_{j=i+1}^{n} c_j \cdot \frac{dq_{j-1}}{dq_{j-2}} \cdot \frac{dq_{j-2}}{dq_{j-3}} \cdots \frac{dq_{i}}{dc_i} & \mbox{(by the chain rule)}\\
= & q_{i-1} + \sum_{j=i+1}^{n} c_j \cdot \prod_{k=i+1}^{j-1} \left( 1-\frac{2c_k}{k-1} \right) \cdot \left( -\frac{2q_{i-1}}{i-1} \right) & \mbox{(by the recursive formula of $q_j$, see Equation~\eqref{eq:qi-recursive})}\\
= & q_{i-1} \cdot \left( 1 - \frac{2}{i-1} \cdot \sum_{j=i+1}^{n}c_j \prod_{k=i+1}^{j-1} \left( 1- \frac{2c_{k}}{k-1} \right) \right).
\end{align*}
Notice that $\frac{df}{dc_i}$ does not depend on $c_i$, i.e., f is a linear function of $c_i$. In particular, it means that the maximum of $f(\vec{c})$ on $\vec{c}\in[0,1]^n$ is achieved at either $c_i=0$, or $c_i=1$ depending on the sign of $\frac{df}{dc_i}$. Note that the maximum of $f$ must be attained by some $\vec{c}$, since $[0,1]^n$ is a compact space and $f$ is a continuous function. Moreover, we observe the following ``monotonicity'' property of the derivatives $\frac{df}{dc_i}$.
\begin{claim}
	\label{cl:derivative}
Let $i\ge 3$. If $\frac{df}{dc_i} \le 0$, then $\frac{df}{dc_{i-1}} < 0$.
\end{claim}
\begin{proof}
Let $x_{i-1} = \sum_{j=i}^{n}c_j \prod_{k=i}^{j-1} \left( 1- \frac{2c_{k}}{k-1} \right)$ and $x_i = \sum_{j=i+1}^{n} c_j\prod_{k=i+1}^{j-1} \left( 1- \frac{2c_{k}}{k-1} \right)$. Thus $x_{i-1} = x_i \cdot \left(1-\frac{2c_i}{i-1}\right) + c_i$.
Then, $\frac{df}{dc_i} \le 0$ iff $1-\frac{2x_i}{i-1}\le 0$ iff $x_i \ge \frac{i-1}{2}$. 

We get:
\begin{multline*}
x_{i-1} = x_i \cdot \left(1-\frac{2c_i}{i-1}\right) + c_i =
x_i + c_i\cdot\left(1-\frac{2x_i}{i-1} \right)
\ge
x_i + 1\cdot \left(1-\frac{2x_i}{i-1} \right)\\
=
x_i\cdot\frac{i-3}{i-1}+1
\ge\frac{i-1}{2}\cdot\frac{i-3}{i-1}+1
=\frac{i-1}{2}>\frac{i-2}{2}.
\end{multline*}
Consequently $\frac{df}{dc_{i-1}} < 0$.
\end{proof}

Let $\vec{c}\in[0,1]^n$ be the vector at which $f(\vec{c})$ attains its maximum. 
Let $\ell$ be the largest index $i\in[n]$ such that $\frac{df}{dc_{i}}\le 0$. Then, $\frac{df}{dc_{i}}> 0$ for all $i>\ell$, and 
by Claim~\ref{cl:derivative}, $\frac{df}{dc_{i}}< 0$ for all $i<\ell$. 
Since $f$ attains its maximum at $\vec{c}$, it must be that $c_i=1$ for all $i>\ell$ and $c_i=0$ for all $i<\ell$. 
We can also assume without loss of generality that $c_{\ell}=0$, as 
$\frac{df}{dc_{\ell}}\le 0$ (if $\frac{df}{dc_{\ell}}= 0$, then $f$ does not depend on $c_{\ell}$). 
We conclude that $c_i=0$ for all $i \leq \ell$ and $c_i=1$ for $i>\ell$.
In other words, the algorithm is deterministic: it matches no vertices up to step $\ell$, and thereafter matches the top two vertices so far whenever possible.


Next, we calculate the value of $q_i$ for all $i$.
\begin{itemize}
	\item For $i \le \ell$, $c_i=0$, and we get $q_i=1+q_{i-1}\cdot(1-\frac{2\cdot c_i}{i-1})=1+q_{i-1}$. Thus, $q_i=i$.
	\item For 
 $i > \ell$, $c_i=1$, and we get $q_{i} = 1+q_{i-1}\cdot(1-\frac{2\cdot c_i}{i-1})=1 + \frac{i-3}{i-1} \cdot q_{i-1}$.
 Multiplying the last equality by $(i-1)(i-2)$ gives
\begin{equation}
\label{eq:diff1}
	(i-1)(i-2)q_{i} - (i-2)(i-3) q_{i-1} = (i-1)(i-2). 
\end{equation}
Summing the LHS of \eqref{eq:diff1} over $j=\ell+1,\ldots,i$ (with $j$ in the role of $i$) gives $(i-1)(i-2)q_{i} - \ell (\ell-1)(\ell-2)$ due to telescopic sum.
Summing the RHS of \eqref{eq:diff1} over $j=\ell+1,\ldots,i$ (with $j$ in the role of $i$) gives $\sum_{j=\ell+1}^{i} (j-1)(j-2) = \frac{i(i-1)(i-2)}{3}-\frac{\ell(\ell-1)(\ell-2)}{3}$.
We get that
$q_i = \frac{i}{3} + \frac{2\ell (\ell-1)(\ell-2)}{3(i-1)(i-2)}$.
\end{itemize}
 

We are now ready to calculate the performance of the algorithm. Using Equation~\eqref{eq:obj}, the fact that $c_i=0$ for all $i\leq \ell$ and $c_i=1$ for all $i>\ell$, and the values of $q_i$ as calculated above gives
\begin{align*}
\binom{n}{2}\cdot \alg = & \sum_{i=2}^{n} q_{i-1}c_i = \sum_{i=\ell}^{n-1} q_i = \sum_{i=\ell}^{n-1} \left( \frac{i}{3} + \frac{2\ell(\ell-1)(\ell-2)}{3(i-1)(i-2)} \right) \\
=& \frac{ \sum_{i=1}^{n-1} i -\sum_{i=1}^{\ell-1} i}{3} +
 \frac{2\ell(\ell-1)(\ell-2)}{3} \sum_{i=\ell}^{n-1}\left[\frac{1}{i-2}-\frac{1}{i-1}\right] \\
= & \frac{n^2-n-\ell^2+\ell}{6} + \frac{2}{3}\ell(\ell-1)(\ell-2)\left(\frac{1}{\ell-2} - \frac{1}{n-2}\right) \\
= & \frac{n^2}{6}+ \frac{\ell^2}{2} - \frac{2\ell^3}{3n}+\left[\frac{\ell-n}{6}-\frac{2\ell}{3}+\frac{2(3\ell^2-2\ell)}{3n}-\frac{4\ell(\ell-1)(\ell-2)}{3n(n-2)} \right]\\
= & n^2 \cdot\left(\frac{1}{6}+\frac{\ell^2}{2n^2} - \frac{2\ell^3}{3n^3} + O\left(\frac{1}{n}\right) \right)
\le \frac{5}{24} n^2 +O(n),
\end{align*}
where the  last equality holds since $0<\ell\le n$; and the last inequality holds since the cubic function
$g(x)\eqdef \frac{1}{6}+\frac{x^2}{2} -\frac{2x^3}{3}$ with $x=\frac{\ell}{n}$ attains its maximum on the interval $x\in[0,1]$ at $x=0.5$, where the maximum value is $\frac{5}{24}$.
Therefore, $\alg \le \frac{5}{12} + O(\frac{1}{n})$, concluding the proof of Theorem~\ref{thm:ordinal-5-12}.
\end{proof}

\section{Secretary Matching with Edge Arrival}
\label{sec:edge_arrival}

In this section we present an algorithm that gives a competitive ratio of $1/4$ for edge arrival. 
Let $e_1, \ldots, e_m$ be the edges in their arrival order. 
Let $E_t=\{e_1, \ldots, e_t\}$ denote the set of edges that arrived up to time $t$, and let $\mu^*_t$ denote the (unique) maximum weighted matching in $G(E_t)$. 

The algorithm (see Algorithm~\ref{alg:edge-arrival}) ignores the first $\lfloor \frac{m}{2} \rfloor$ edges (exploration phase). 
Then, in every round $t$, upon the arrival of edge $e_t=uv$, it computes the probability that both $u$ and $v$ are available 
(the probability is taken with respect to the random arrival order of the edges $E_{t-1}$ and the random choices of the algorithm in steps $1$ to $t-1$); 
denote this probability by $x_t$.
It then finds the maximum weighted matching $\mu^*_t$ in the graph induced by $E_t$. 
If $e_t \in \mu^*_t$, the algorithm matches $e_t$ with probability $\frac{\alpha_t}{x_t}$, where $\alpha_t$ is given by the following formula~\eqref{eq:alphatdef}:
\begin{equation}
\alpha_t = 
\begin{cases}
0 & \text{if } t\leq \frac{m}{2}\\
1-2\sum_{i=1}^{t-1} \frac{\alpha_i}{i} & \text{if } t>\frac{m}{2}
\end{cases}
\label{eq:alphatdef}
\end{equation}

\begin{algorithm}
	\caption{$1/4$-competitive algorithm for matching secretary with edge arrival}
	\label{alg:edge-arrival}
	\begin{algorithmic}[1]
		\STATE Let $e_1,\ldots,e_m$ be the edges given in their arrival order
		\STATE $A=V$
		\quad\quad\quad\quad\quad\quad\quad\quad
		\tcp{$A$ is the set of available vertices}
		\STATE $\mu=\emptyset$
		\quad\quad\quad\quad\quad\quad\quad\quad
		\tcp{$\mu$ is the returned matching}
		\FOR{$t \in\{\lfloor\frac{m}{2}\rfloor+1,...,m\}$} 
		\STATE Let $e_t=uv$ be the edge arriving at time $t$
		\STATE Let $x_t$ be the probability that $e_t$ is available (i.e., $u,v \in A$) 
		\STATE Let $\mu^*_t$ be the maximum weighted matching in $G(E_t)$
		\IF{$e_t \in \mu^*_t$  and $u,v \in A$}
		\STATE With probability $\frac{\alpha_t}{x_t}$
		\STATE \ \ Add $e_t$ to $\mu$
		\STATE \ \ Remove $u$ and $v$ from $A$
		\ENDIF  
		\ENDFOR
		\STATE Return matching $\mu$
	\end{algorithmic}
\end{algorithm} 
The algorithm ensures that every given edge $e_t=uv$ is matched with a certain probability $\alpha_t$ given by \eqref{eq:alphatdef}, whenever $e_t$ is in the current maximum matching. This allows us to conveniently estimate the expected contribution of the maximum matching at time $t$ and compare it to the maximum matching in the whole graph. Before doing that, we need to prove that the algorithm is well defined, i.e., that $x_t \geq \alpha_t$ for every $t$.
Note that $x_t$ is the probability that $e_t$ is available, given a random order of edges in $E_t$ (the edges arriving up to time $t$) that ends with $e_t$. As such, it depends on $e_t$, and the set $E_t$, but {\em not} on the order of edges within $E_t$ (except for $e_t$ being the $t$th edge). 
Recall that ``$\matbefore{u}{t}$" denotes the event that $u$ becomes unavailable before time $t$, and ``$\matat{e}{t}$" denotes the event that edge $e$ is matched exactly at time $t$. 
\begin{lemma}
	\label{lem:well-defined}
For every time $t$, vertices $u,v$ and set of edges $Q$ of size $t-1$, given that $e_t=uv$ and $E_{t-1}=Q$, it holds that $x_t \geq \alpha_t $.
\end{lemma}
\begin{proof}

We prove by induction on $t$. For the base case, where $t=\lfloor \frac{m}{2} \rfloor + 1$, the statement holds trivially since we have not matched anything yet, i.e., $x_t=1$.
Next, fix $t$ and suppose the statement holds for all $t' \le t-1$. Let $e_t=uv$ be the edge arriving at time $t$ and the let $Q=E_{t-1}$ be the set of arrived edges before time $t$.

For simplicity, in the remainder of the proof, we omit the given $e_t=uv$ and $E_{t-1} = Q$ in all probabilities and indicator expressions.

It holds that
$$
\Prx{\matbefore{u}{t}}   =   \sum_{\substack{e\ni u, \\ e\in Q} }  \sum_{i=1}^{t-1} \Prx{\aat{e}{i}} \cdot \Prx{\matat{e}{i}  \mid \aat{e}{i}},
$$
Clearly, $\Prx{\aat{e}{i}}=\frac{1}{t-1}$ due to the random arrival order. To calculate $ \Prx{\matat{e}{i}  \mid \aat{e}{i}}$, note that
the edge $e_i=uv'$ is matched at time $i$ if (i) $e_i$ belongs to $\mu^*_i$, and (ii) both $u$ and $v'$ are available. Under (i) and (ii), $e_i$ is matched with probability $\alpha_i/x_i$ by the induction hypothesis. 
Since (i) and (ii) are independent and (ii) occurs with probability $x_i$, we get:
 \begin{eqnarray}
 \Prx{\matbefore{u}{t}} & = & \sum_{\substack{e\ni u, \\ e\in Q} }  \sum_{i=1}^{t-1} \frac{1}{t-1} \cdot \Prx{e \in \mu^*_i \mid \aat{e}{i}} \cdot x_i \cdot \frac{\alpha_i}{x_i} \nonumber\\
  & = & \frac{1}{t-1}  \sum_{i=1}^{t-1} \alpha_i \sum_{\substack{e\ni u, \\ e\in Q}  }    \Prx{e \in \mu^*_i \mid \aat{e}{i}}, 
 \label{eq:u-is-unmatched}
\end{eqnarray}
where the last equality is obtained by changing the order of summation.
We next prove that 
\begin{equation}
\label{eq:edge-ub}
	\sum_{\substack{e\ni u, \\ e\in Q} }\Prx{e \in \mu^*_i \mid \aat{e}{i}} \leq \frac{t-1}{i}.
\end{equation}
It holds that
$$
\sum_{\substack{e\ni u, \\ e\in Q}  } \Prx{e \in \mu^*_i \mid \aat{e}{i}} = \sum_{\substack{e\ni u, \\ e\in Q} }   
\sum_{\substack{S \subseteq Q,\\ e\in S, |S|=i}} 
\indic[e \in \mu^*_i \mid \aat{e}{i}, \sfe{S}{i}] \cdot \Prx{\sfe{S}{i} \mid \aat{e}{i}}. 
$$
Since the arrival order is chosen uniformly at random, $\forall S\subset Q: |S|=i$ we have $\Prx{\sfe{S}{i} \mid \aat{e}{i}} = \frac{1}{{{t-2}\choose{i-1}}}$, which can be written as $\frac{1}{{{t-1}\choose{i}}} \cdot \frac{t-1}{i}$. By changing the order of summation we get
\begin{multline*} 
\sum_{\substack{e\ni u, \\ e\in Q} } \Prx{e \in \mu^*_i \mid \aat{e}{i}}  = 
 \sum_{\substack{S \subseteq Q,\\ |S|=i}}   \sum_{\substack{e\ni u,\\ e \in S}} \indic[e \in \mu^*_i \mid  \sfe{S}{i}] \cdot \frac{1}{{{t-1}\choose{i}}} \cdot \frac{t-1}{i} \\
 =   \frac{t-1}{i} \sum_{\substack{S \subseteq Q,\\ |S|=i}}   \sum_{\substack{e\ni u,\\ e \in S}} \indic[e \in \mu^*_i \mid \sfe{S}{i}] \cdot \Prx{\sfe{S}{i}} 
 \leq   \frac{t-1}{i} \sum_{\substack{S \subseteq Q,\\ |S|=i}}    \Prx{\sfe{S}{i}} =    \frac{t-1}{i}. 
\end{multline*}
The second equality holds since $\Prx{\sfe{S}{i}} = \frac{1}{{{t-1}\choose{i}}}$, and the inequality follows by observing that $\sum_{ e\ni u, e \in S} \indic[e \in \mu^*_i \mid \sfe{S}{i}] \leq 1$ since the events $e \in \mu^*_i$ are disjoint for different $e\ni u$. The last equality follows since $\sum_{S \subseteq Q, |S|=i}\Prx{\sfe{S}{i}}=1$, as a partition into all possible realizations of $E_i$. This concludes the proof of Equation~\eqref{eq:edge-ub}.
We combine \eqref{eq:u-is-unmatched} and \eqref{eq:edge-ub} and get the following probability bound that $u$ (and similarly $v$) is matched before time $t$:
\[
\Prx{\matbefore{u}{t}} \leq \sum_{i=1}^{t-1} \frac{\alpha_i}{i}\quad, \quad
\Prx{\matbefore{v}{t}} \leq \sum_{i=1}^{t-1} \frac{\alpha_i}{i}
\]

Applying the union bound to the events $\matbefore{u}{t}$ and $\matbefore{v}{t}$ the probability that both $u$ and $v$ are available upon the arrival of $e_t$ is
\begin{equation*}
x_t = \Prx{\avail{u,v}{t}}  \geq 1- 2 \sum_{i=1}^{t-1} \frac{\alpha_i}{i} 
\stackrel{\eqref{eq:alphatdef}}{=}
\alpha_t.
\end{equation*}
This concludes the proof of Lemma~\ref{lem:well-defined}.
\end{proof}
We are now ready to prove that Algorithm~\ref{alg:edge-arrival} is $\frac{1}{4}$-competitive.
\begin{theorem}
\label{thm:edge_arrival}
	Algorithm~\ref{alg:edge-arrival} has a competitive ratio of $\frac{1}{4}$.
\end{theorem}

\begin{proof}
First, as the set of the first $t$ edges $E_t$ is chosen uniformly at random, the maximum matching $\mu^*_t$ in $E_t$ 
has the expected total weight greater than or equal to the expected weight of the $E_t$ edges in the optimal matching $\mu^*$
\begin{equation}
\Ex{\sum_{e\in\mu^*_t}w_e} \geq \frac{t}{m}\cdot \sum_{e\in\mu^*}w_e 
\label{eq:opti}.
\end{equation}

Next, we prove by induction that for every $t>\frac{m}{2}$:
\begin{equation}
\alpha_t=\frac{\lfloor \frac{m}{2} \rfloor \cdot \lfloor \frac{m-2}{2} \rfloor  }{(t-1)(t-2)}. 
\label{eq:alphat}
\end{equation}
For $t = \lfloor \frac{m}{2} \rfloor+1$, $\alpha_t$ is indeed $1$.
For $t>\lfloor \frac{m}{2} \rfloor+1$, $$\alpha_t = 1- 2 \sum_{i=1}^{t-1} \frac{\alpha_i}{i}= \alpha_{t-1} - 2 \cdot \frac{\alpha_{t-1}}{t-1} = \frac{t-3}{t-1} \cdot \alpha_{t-1}.$$ 
The induction hypothesis now implies that $\frac{t-3}{t-1} \cdot \alpha_{t-1} = \frac{t-3}{t-1} \cdot \frac{\lfloor \frac{m}{2} \rfloor \cdot\lfloor \frac{m-2}{2} \rfloor  }{(t-2)(t-3)} = \frac{\lfloor \frac{m}{2} \rfloor \cdot\lfloor \frac{m-2}{2} \rfloor  }{(t-1)(t-2)}$, concluding~\eqref{eq:alphat}.

We are now ready to establish the competitive ratio of Algorithm~\ref{alg:edge-arrival}. Write $\mu$ for the matching returned by Algorithm~\ref{alg:edge-arrival}.
Edge $e_t=uv$ is matched in round $t$ if: (i) it belongs to $\mu^*_t$, and (ii) both $u$ and $v$ are available (this happens with probability $x_t$). Under these two events, $e_t$ is matched with probability $\alpha_t/x_t$. Since events (i) and (ii) are independent, we get that
$$
\Ex{\sum_{e\in\mu} w_e} = \sum_{t=\lfloor\frac{m}{2}\rfloor+1}^{m} \sum_{\substack{S \subseteq E,\\ |S|=t}} \sum_{e\in S} \Ex{w_e \cdot \indic[e \in \mu^*_t ] \mid  e=e_t, E_t=S } \cdot \Prx{E_t =S, e=e_t} \cdot \alpha_t.
$$
Observe that $\sum_{e\in S} \Ex{w_e \cdot \indic[e \in \mu^*_t ] \mid  e=e_t, E_t=S} = \Ex{\sum_{e\in\mu^*_t} w_e  \mid   E_t=S }$ and also that $\Prx{E_t =S, e=e_t} = \Prx{E_t=S}\cdot \frac{1}{t}$. We get

\begin{eqnarray}
\Ex{\sum_{e\in\mu} w_e}
	& = &  \sum_{t=\lfloor\frac{m}{2}\rfloor+1}^{m} \sum_{\substack{S \subseteq E,\\ |S|=t}} \Ex{\sum_{e\in\mu^*_t}w_e  ~\Big\vert~   E_t=S } \cdot \Prx{E_t=S}\cdot\frac{1}{t} \cdot \alpha_t \nonumber \\	
& = &  
	\sum_{t=\lfloor\frac{m}{2}\rfloor+1}^{m}  \Ex{\sum_{e\in\mu^*_t}w_e }  \cdot\frac{\alpha_t}{t}
	\stackrel{\eqref{eq:alphat}}{=}
	\sum_{t=\lfloor\frac{m}{2}\rfloor+1}^{m}  \Ex{\sum_{e\in\mu^*_t}w_e }  \cdot\frac{1}{t}  \cdot  \frac{\lfloor \frac{m}{2} \rfloor \cdot\lfloor \frac{m-2}{2} \rfloor  }{(t-1)(t-2)}
	\nonumber \\
	& \stackrel{\eqref{eq:opti}}{\geq} &  \sum_{t=\lfloor\frac{m}{2}\rfloor+1}^{m}  \frac{t\cdot \sum_{e\in\mu^*}w_e}{m} \cdot\frac{1}{t}  \cdot \frac{\lfloor \frac{m}{2} \rfloor \cdot \lfloor \frac{m-2}{2} \rfloor  }{(t-1)(t-2)} \nonumber \\
	& = &  \left\lfloor \frac{m}{2} \right\rfloor  \cdot \left\lfloor \frac{m-2}{2} \right\rfloor \cdot \frac{\sum_{e\in\mu^*}w_e  }{m} \cdot \sum_{t=\lfloor\frac{m}{2}\rfloor+1}^{m}   \frac{1}{(t-1)(t-2)} 
		\ge \frac{1}{4} \sum_{e\in\mu^*}w_e 
\label{eq:teleskopic_sum}
\end{eqnarray}
The last inequality follows by observing that $\frac{1}{(t-1)(t-2)} = \frac{1}{t-2} - \frac{1}{t-1}$, thus the sum telescopes to 
$\frac{1}{\lfloor\frac{m}{2}\rfloor-1}-\frac{1}{m-1}$; we have $\left\lfloor \frac{m}{2} \right\rfloor  \cdot \left\lfloor \frac{m-2}{2}\right\rfloor \cdot\frac{1}{m}\left(\frac{1}{\lfloor\frac{m}{2}\rfloor-1}-\frac{1}{m-1}\right)=\frac{\lfloor \frac{m}{2}\rfloor}{m}
\left(1-\frac{\lfloor \frac{m-2}{2}\rfloor}{m-1}\right)>\frac{1}{2}\left(1-\frac{1}{2}\right)=1/4$. I.e., the coefficient at $\sum_{e\in\mu^*}w_e$ is  at least $\frac{1}{4}$. This concludes the proof of the theorem.
\end{proof}

\noindent {\bf Remark:}
It is quite straightforward to implement  Algorithm~\ref{alg:edge-arrival} in {\em exponential} (in $m$) time using Monte Carlo simulations: one simply needs to compute all $x_t=x_t(S)$ depending on the set of visible edges $S\subset E$ (i.e., edges with known weights). We can easily compute all $\alpha_t$ using the simple explicit formula \eqref{eq:alphat}. Unfortunately, we do not know how to efficiently compute or estimate $x_t$ in subexponential time (to estimate the probability that both ends of $e_t$ are available, we need to predict at each step $i<t$ what  Algorithm~\ref{alg:edge-arrival} would do for a random set of edges $S\subset E_t$). 
Thus, our result in Theorem~\ref{thm:edge_arrival} can be seen as an information-theoretic result. It remains an interesting open problem whether there is a poly-time online algorithm that matches this bound of $1/4$.

	\bibliographystyle{alpha}
	\bibliography{prophet-matching}
	
	\newpage
	\appendix

\section{Generalization to Hypergraphs}
\label{sec:hypergraph}

In this section we generalize Algorithm~\ref{alg:edge-arrival} to the online bipartite hypergraph secretary matching problem studied by \citet{KorulaP09} and \citet{KesselheimRTV13}. 

Let $H= (L \cup R, E)$ be the underlying edge-weighted $(d+1)$-hypergraph. Each edge in $E$ has the form $(v,S)$ where $v \in L, S \subseteq R$ and $|S| \le d$. All vertices in $R$ are given in advance and the vertices in $L$ arrive online uniformly at random. We assume $|L|=m$, and $m$ is known in advance. Upon the arrival of vertex $v$, all its incident hyperedges are revealed with the corresponding weights.
The edge arrival model studied in Section~\ref{sec:edge_arrival} can be viewed as a special case of the online bipartite $3$-hypergraph matching. Specifically, for the underlying graph $G=(V,E_G)$ of the edge arrival problem we construct a hypergraph in which $R=V$ and each vertex in $L$ corresponds to an edge in $E_G$. That is, each vertex $\ell\in L$ corresponding to the edge $e=(uv)\in E_G$ has only one incident hyperedge $\{\ell,u,v\}$ in the hypergraph. 

Let $L_t=\{\ell_1, \ldots, \ell_t\}$ denote the set of vertices that arrived up to time $t$, and let $\mu^*_t$ denote the (unique) maximum weighted matching in $H(L_t \cup R)$. 
The algorithm (see Algorithm~\ref{alg:hypergraph}) ignores the first $\lfloor f_d \cdot m \rfloor$ vertices (exploration phase), where $f_d = 1 / d^{\frac{1}{d-1}}$. 
Then, in every round $t$, upon the arrival of vertex $\ell_t$, we first find the maximum weighted matching $\mu^*_t$ in the graph induced by $L_t \cup R$. Let $e_t$ be the incident edge of $\ell_t$ in $\mu^*_t$. If $\ell_t$ is not matched in $\mu^*_t$, let $e_t$ be a null edge for notation simplicity. We compute the probability that edge $e_t$ is available (the probability is taken with respect to the random arrival order of the vertices $L_{t-1}$ and random choices of the algorithm in steps $1$ to $t-1$); 
denote this probability by $x_t$.
Then, the algorithm matches $e_t$ with probability $\frac{\alpha_t}{x_t}$, where $\alpha_t$ is given by the following formula~\eqref{eq:hyper_alphatdef}:
\begin{equation}
\alpha_t = 
\begin{cases}
0 & \text{if } t\leq f_d \cdot m\\
1- d \cdot \sum_{i=1}^{t-1} \frac{\alpha_i}{i} & \text{if } t > f_d \cdot m
\end{cases}
\label{eq:hyper_alphatdef}
\end{equation}

\begin{algorithm}
	\caption{Algorithm for online bipartite hypergraph matching secretary}
	\label{alg:hypergraph}
	\begin{algorithmic}[1]
		\STATE Let $\ell_1,\ldots,\ell_m$ be the vertices of the left side $L$ given in their arrival order
		\STATE $A=R$
		\quad\quad\quad\quad\quad\quad\quad\quad
		\tcp{$A$ is the set of available vertices}
		\STATE $\mu=\emptyset$
		\quad\quad\quad\quad\quad\quad\quad\quad
		\tcp{$\mu$ is the returned matching}
		\FOR{$t \in\{ \lfloor f_d \cdot m \rfloor +1,...,m\}$} 
		\STATE Let $\mu^*_t$ be the maximum weighted matching in $H(L_t \cup R)$
		\STATE Let $e_t=(\ell_t, S_t)$ be the incident edge of $\ell_t$ in $\mu^*_t$
		\quad \tcp{If $\ell_t$ is not matched in $\mu^*_t$, let $S_t = \emptyset$}
		\STATE Let $x_t$ be the probability that $e_t$ is available (i.e., $S_t \subseteq A$) 

		\IF{$S_t \subseteq A$}
		\STATE With probability $\frac{\alpha_t}{x_t}$
		\STATE \ \ Add $e_t$ to $\mu$
		\STATE \ \ Remove $S_t$ from $A$
		\ENDIF  
		\ENDFOR
		\STATE Return matching $\mu$
	\end{algorithmic}
\end{algorithm} 

Observe that the edge $e_t$ we consider to take at time $t$ only depends on the set of arrived vertices $L_t$ at time $t$ and the last arriving vertex $\ell_t$. 
The algorithm ensures that every edge $e_t$ is matched with a certain probability $\alpha_t$ given by \eqref{eq:hyper_alphatdef}. This allows us to conveniently estimate the expected contribution of the maximum matching at time $t$ and compare it to the maximum matching in the whole graph. Before doing that, we need to prove that the algorithm is well defined, i.e., that $x_t \geq \alpha_t$ for every $t$.

Note that $x_t$ is the probability that $e_t$ is available, given a random order of vertices in $L_t$ (the vertices arriving up to time $t$) that ends with $\ell_t$. As such, the probability depends on $\ell_t$, and the set $L_t$, but {\em not} on the order of vertices within $L_t$ (except for $\ell_t$ being the $t$-th vertex). 

Recall that ``$\matbefore{u}{t}$" denotes the event that $u\in R$ becomes unavailable before time $t$, and ``$\matat{e}{t}$" denotes the event that edge $e$ is matched exactly at time $t$. 

\begin{lemma}
	\label{lem:hyper_well-defined}
	For every time $t$, vertex $v\in L$, and a set of vertices $Q$, if $\ell_t=v$ and $L_{t-1}=Q$, then $x_t \geq \alpha_t $.
\end{lemma}


\begin{proof}
	The proof proceeds by induction on $t$. For the base case, where $t=\lfloor f_d \cdot m \rfloor + 1$, the statement holds trivially since we have not matched anything yet, i.e., $x_t=1$.
	Next, fix $t$ and suppose the statement holds for all $t' \le t-1$. Let $v=\ell_t$ be the vertex arriving at time $t$ and let $Q = L_{t-1}$ be the set of arrived vertices up to time $t$. Observe that given $v=\ell_t, Q=L_{t-1}$, the edge $e_t = (v, S)$ is fixed.
	For simplicity, we omit conditioning on $\ell_t=v, L_{t-1} = Q$ in all probabilities and indicator expressions in the remainder of the lemma's proof. 
	For each vertex $u \in R$, it holds that
	\begin{multline}
	\Prx{\matbefore{u}{t}}  =  \sum_{i=1}^{t-1} \sum_{\substack{Q_i \subseteq Q\\ |Q_i| = i}} \sum_{z_i \in Q_i} \Prx{\ell_i = z_i, L_i = Q_i} \cdot \indic[u \in e_i \mid \ell_i = z_i, L_i=Q_i] \\
	\cdot \Prx{\matat{e_i}{i}  \mid \ell_i = z_i, L_i=Q_i}.
	\end{multline}
	Notice that $\Prx{\matat{e_i}{i}  \mid \ell_i=z_i, L_i=Q_i} = \alpha_i$, according to the design of our algorithm and the induction hypothesis, and that $\Prx{\ell_i=z_i, L_i=Q_i} = \frac{1}{i} \Prx{L_i=Q_i}$ due to the random arrival order of the vertices. We get
	\begin{eqnarray}
		\Prx{\matbefore{u}{t}} & = &  \sum_{i=1}^{t-1} \frac{\alpha_i}{i} \cdot \sum_{\substack{Q_i \subseteq Q\\ |Q_i| = i}} \Prx{L_i = Q_i} \sum_{z_i \in Q_i} \indic[u \in e_i \mid \ell_i = z_i , L_i=Q_i]. \nonumber 
		\label{eq:hyper_u-is-unmatched}
	\end{eqnarray}
The maximum matching $\mu^*_i$ is fixed for a given $Q_i$ and $u \in e_i$ if and only if (i) $u$ is matched in $\mu^*_i$ and (ii) its corresponding online vertex arrives at time $i$. That is, $\sum_{z_i \in Q_i } \indic[u \in e_i \mid \ell_i = z_i , L_i=Q_i] \le 1$.
	Consequently, we have that
	\begin{eqnarray}
	\Prx{\matbefore{u}{t}} & \le &  \sum_{i=1}^{t-1} \frac{\alpha_i}{i} \cdot \sum_{\substack{Q_i \subseteq Q\\ |Q_i| = i}} \Prx{L_i = Q_i} \cdot 1 = \sum_{i=1}^{t-1} \frac{\alpha_i}{i}.
	\end{eqnarray}
	Finally, since $S_t$ contains at most $d$ vertices, applying the union bound to the events $\matbefore{u}{t}$ for all $u \in S_t$, we have that the probability that $e_t$ is available is
	\begin{equation*}
		x_t = \Prx{\avail{e_t}{t}}  \geq 1- d \sum_{i=1}^{t-1} \frac{\alpha_i}{i} 
		\stackrel{\eqref{eq:hyper_alphatdef}}{=}
		\alpha_t.
	\end{equation*}
	This concludes the proof of Lemma~\ref{lem:hyper_well-defined}.
\end{proof}

We are now ready to conclude the competitive analysis of Algorithm~\ref{alg:hypergraph}. Our competitive ratio has the same asymptotic order $\Omega(\frac{1}{d})$ as the previous best bound of $\frac{1}{ed}$ by \citet{KesselheimRTV13}, but the constant factor of $e$ improves when $d$ goes to infinity.  We assume without loss of generality that the number of vertices $m$ in $L$ is large. Indeed, we can slightly modify Algorithm~\ref{alg:hypergraph} by adding a number of dummy vertices with no edges to $R$. 

\begin{theorem}
	\label{thm:hypergraph}
	Algorithm~\ref{alg:hypergraph} has a competitive ratio of $1/d^{\frac{d}{d-1}}$.
\end{theorem}

\begin{proof}  As the set of the first $t$ vertices $L_t$ is chosen uniformly at random, the expected total weight of the maximum matching $\mu^*_t$ is greater than or equal to the expected weight of those edges of $L_t$ in the optimal matching $\mu^*$, i.e.,
	\begin{equation}
		\Ex{\sum_{e\in\mu^*_t}w_e} \geq \frac{t}{m}\cdot \sum_{e\in\mu^*}w_e 
		\label{eq:hyper_opti}.
	\end{equation}
	
	Next, we prove by induction on $t$ that for every $t> \lfloor f_d \cdot m \rfloor $:
	\begin{equation}
		\alpha_t= \prod_{i=1}^d \frac{ \lfloor f_d \cdot m \rfloor + 1 -i}{t-i}. \label{eq:hyper_alphat}
	\end{equation}
	For $t = \lfloor f_d \cdot m \rfloor +1$, $\alpha_t$ is indeed $1$.
	For $t> \lfloor f_d \cdot m \rfloor + 1$, $$\alpha_t = 1- d \sum_{i=1}^{t-1} \frac{\alpha_i}{i}= \alpha_{t-1} - d \cdot \frac{\alpha_{t-1}}{t-1} = \frac{t-d-1}{t-1} \cdot \alpha_{t-1}=\frac{t-d-1}{t-1}
	\prod_{i=1}^d \frac{ \lfloor f_d \cdot m \rfloor + 1 -i}{t-1-i}
	.$$ 
	The induction hypothesis now implies \eqref{eq:hyper_alphat}.

	We are now ready to establish the competitive ratio of Algorithm~\ref{alg:hypergraph}. Write $\mu$ for the matching returned by Algorithm~\ref{alg:hypergraph}.
	Edge $e_t=(v,S)$ is matched in round $t$ if all vertices in $S$ are available (this happens with probability $x_t$). Under this event, $e_t$ is matched with probability $\alpha_t/x_t$. Therefore,
	$$
	\Ex{\sum_{e\in\mu} w_e} = \sum_{t=\lfloor f_d \cdot m \rfloor+1}^{m} \sum_{\substack{V \subseteq L,\\ |V|=t}} \sum_{v \in V} \Ex{w_{e_t} \mid  \ell_t=v, L_t=V } \cdot \Prx{\ell_t = v, L_t = V} \cdot \alpha_t.
	$$
	Observe that $\sum_{v\in V} \Ex{w_{e_t} \mid  \ell_t=v, L_t=V} = \Ex{\sum_{e\in\mu^*_t} w_e  \mid   L_t=V }$ and also that $\Prx{\ell_t = v, L_t =V} = \Prx{L_t=V}\cdot \frac{1}{t}$. We get
	
	\begin{eqnarray}
		\Ex{\sum_{e\in\mu} w_e}
		& = &  \sum_{t=\lfloor f_d \cdot m \rfloor+1}^{m} \sum_{\substack{V \subseteq L,\\ |V|=t}} \Ex{\sum_{e\in\mu^*_t}w_e  ~\Big\vert~   L_t=V } \cdot \Prx{L_t=V}\cdot\frac{1}{t} \cdot \alpha_t \nonumber \\	
		& = &  
		\sum_{t=\lfloor f_d \cdot m\rfloor+1}^{m}  \Ex{\sum_{e\in\mu^*_t}w_e }  \cdot\frac{\alpha_t}{t}
		\stackrel{\eqref{eq:hyper_alphat}}{=}
		\sum_{t=\lfloor f_d \cdot m \rfloor+1}^{m}  \Ex{\sum_{e\in\mu^*_t}w_e }  \cdot\frac{1}{t}  \cdot  \prod_{i=1}^d \frac{ \lfloor f_d \cdot m \rfloor + 1 -i}{t-i}
		\nonumber \\
		& \stackrel{\eqref{eq:hyper_opti}}{\geq} &  \sum_{t=\lfloor f_d \cdot m \rfloor+1}^{m}  \frac{t\cdot \sum_{e\in\mu^*}w_e}{m} \cdot\frac{1}{t} \cdot \prod_{i=1}^d \frac{ \lfloor f_d \cdot m \rfloor + 1 -i}{t-i} \nonumber \\
		& = & \prod_{i=1}^{d} \left( \lfloor f_d \cdot m\rfloor +1-i \right) \cdot \frac{\sum_{e\in\mu^*}w_e  }{m} \cdot \sum_{t=\lfloor f_d \cdot m \rfloor+1}^{m} \prod_{i=1}^{d}  \frac{1}{t-i}. 
		\label{eq:hyper_teleskopic_sum}
	\end{eqnarray}

Observe that 
\[
\prod_{i=1}^{d}  \frac{1}{t-i} = \frac{1}{d-1} \left( \prod_{i=1}^{d-1} \frac{1}{t-i-1} - \prod_{i=1}^{d-1} \frac{1}{t-i} \right).
\] 
The summation over $t$ in Equation~\eqref{eq:hyper_teleskopic_sum} telescopes to 
\[
 \sum_{t=\lfloor f_d \cdot m \rfloor+1}^{m} \prod_{i=1}^{d}  \frac{1}{t-i} = \frac{1}{d-1} \cdot \left( \prod_{i=1}^{d-1} \frac{1}{\lfloor f_d \cdot m \rfloor - i}  - \prod_{i=1}^{d-1} \frac{1}{m - i} \right).
\]	
Thus, we have 
\begin{align*}
& \frac{1}{m} \cdot \prod_{i=1}^{d} \left( \lfloor f_d\cdot m\rfloor +1-i \right) \cdot \sum_{t=\lfloor f_d \cdot m \rfloor+1}^{m} \prod_{i=1}^{d}  \frac{1}{t-i} \\
= & \frac{1}{m} \cdot \prod_{i=1}^{d} \left( \lfloor f_d\cdot m\rfloor +1-i \right) \cdot \frac{1}{d-1} \cdot \left( \prod_{i=1}^{d-1} \frac{1}{\lfloor f_d \cdot m \rfloor - i}  - \prod_{i=1}^{d-1} \frac{1}{m - i} \right) \\
= & \frac{1}{d-1} \cdot \left( \frac{\lfloor f_d \cdot m \rfloor}{m}  - \prod_{i=1}^{d} \frac{\lfloor f_d \cdot m \rfloor + 1 - i}{m + 1 -i} \right) \\
= & \frac{1+o(1)}{d-1} \left( f_d - f_d^d \right) =(1+o(1))/d^{\frac{d}{d-1}},
\end{align*}
where to get the second to the last equality we used that $\frac{\lfloor f_d \cdot m \rfloor - i}{m-i}\approx f_d$ for fixed $i$ and $f_d$ as $m$ goes to infinity; the last equation follows from the definition of $f_d = 1/d^{\frac{1}{d-1}}$.
To sum up, the coefficient of $\sum_{e\in\mu^*}w_e$ in Equation~\eqref{eq:hyper_teleskopic_sum} is at least $1/d^{\frac{d}{d-1}}$; this concludes the proof of the theorem.
\end{proof}

%

\end{document}